\documentclass{siamart190516}

\usepackage{graphicx}
\usepackage{latexsym,color}
\usepackage{amssymb}
\usepackage{amsfonts}
\usepackage{amsmath}
\usepackage{comment}
\usepackage[sort&compress,numbers]{natbib}

\newtheorem{thm}{Theorem}[section]

\newtheorem{lem}[thm]{Lemma}
\newtheorem{rem}[thm]{Remark}

\newcommand{\be}{\begin{equation}}
\newcommand{\ee}{\end{equation}}

\newcommand{\delay}{{D}} % delay
\newcommand{\band}{{\mathcal{S}}} % Banded matrices
\newcommand{\toeplitz}{{T}} % Toeplitz matrices
\begin{document}

\title{Explicit Computations for Delayed Semistatic Hedging}
\author{Yan Dolinsky\footnotemark[1] \, and Or Zuk\footnotemark[2]}
\date{\today}
\markboth{Y.Dolinsky and O.Zuk}{Explicit Computations for Delayed Semistatic Hedging}
\maketitle
\renewcommand{\thefootnote}{\fnsymbol{footnote}}
\footnotetext[1]
{Department of Statistics, Hebrew University of Jerusalem.
 \email{yan.dolinsky@mail.huji.ac.il}.}
\footnotetext[2]{Department of Statistics, Hebrew University of Jerusalem.
\email{or.zuk@mail.huji.ac.il}.}
\footnotetext[3] {YD supported by the ISF grant 230/21 and OZ supported by the ISF grant 2392/22.}

\renewcommand{\thefootnote}{\arabic{footnote}}

\pagenumbering{arabic}

\begin{abstract}
In this work we 
consider the exponential utility maximization problem in the framework of semistatic hedging.
In addition to the usual setting considered in Mathematical Finance, 
we also consider an investor who is informed about 
the risky asset’s price changes with a delay. 
In the case where the stock increments are i.i.d.
and normally distributed 
we compute explicitly the value of the problem and the 
corresponding optimal hedging strategy in a discrete and continuous-time setting. 
In discrete time, our approach is based on duality theory and tools from Linear Algebra
which are related to banded matrices and Toeplitz matrices. 
Next, we study an analogous continuous-time model where it is possible to trade at fixed equally spaced times. In this model, we compute the scaling limit for both the trading strategy and the achieved value, as the frequency of the trading events increases. 
Finally, we prove that this scaling limit coincides with the continuous-time hedging problem. 
\end{abstract}
\vspace{10pt}


\vspace{10pt}
\section{Introduction} \label{sec:intro}
We consider a financial market having one risky asset with i.i.d. increments which are 
normally distributed. 
We assume that in addition to trading the stock, the investor
is allowed to take static positions in vanilla options written on the
underlying asset at the maturity date. The prices of the vanilla options are initially known. 
The financial motivation for this
assumption is that vanilla options such as call options are liquid and hence should
be treated as primary assets whose prices are given in the market.

The ``usual'' setting of semistatic hedging 
(see for instance, \cite{hobson1998robust, cox2011robust,dolinsky2014martingale, bouchard2015arbitrage, acciaio2016model, bayraktar2016arbitrage, fahim2016model, guo2017tightness, bartl2019exponential} and the references there in) assumes that the investor 
makes her decision with full access to the prices of the assets.
In this work we consider the case where the investment decisions
may be based only on delayed information, and the actual present asset price is
unknown at the time of decision making. This corresponds to the case where there
is a time delay in receiving market information (or in applying it), which causes
the trader’s filtration to be delayed with respect to the price filtration.
In general, optimal investment problems with delayed filtrations are non Markovian, and so 
the computation
of prices and hedges becomes rather difficult as documented by previous work in the field 
(see \cite{schweizer1994risk, Pham:01, KX:07, mania2008mean, schweizer2018dynamic, saporito2019stochastic, BD:2021, cartea2023optimal}).

In this paper we continue developing and sharpening the approach from 
our recent note \cite{dolinsky2023exponential}
which provides a connection between exponential utility maximization
in the Gaussian framework and 
matrix decompositions. Our first result (Theorem \ref{thm1.1}) establishes an explicit solution 
to the delayed semistatic hedging problem in a discrete-time model with  i.i.d. Gaussian increments.
This result is formulated in Section \ref{sec:2}.
The proof of Theorem \ref{thm1.1} (given in section \ref{sec:3}) is based on the dual approach to the utility maximization problem 
and tools from Linear Algebra which are related to banded matrices 
(see \cite{BF:81}) and
Toeplitz matrices (see \cite{rodman1992inversion}). 

More precisely, we formulate a general verification result (Lemma \ref{lem2.1}) for the 
exponential utility maximization in the semistatic hedging framework. 
In the case where the increments are i.i.d. and normal, finding a solution that satisfies this 
verification result is equivalent to finding a matrix $A$ with three properties: First, $A^{-1}$ is banded, this is closely related to the analog of martingale measures in the setup of hedging with delay.
Second, the sum of the elements of $A$ is a given number which comes from the market 
modelling of the (normal) distribution of the risky asset at the maturity date. 
Third, $A-I$ admits a decomposition as a sum of a constant matrix 
and a banded matrix. We show that the (unique) matrix $A$ satisfying all three properties is a Toeplitz matrix and compute it explicitly. Both the trading strategy and the obtained value of the investor are given in terms of the matrix $A$. 

In Section \ref{sec:4} we treat the Bachelier model. 
We assume 
that the risky asset is 
observed with a constant delay $H>0$.
First, in Section \ref{section4.1} we consider the discrete-time framework  where dynamic hedging is done on an equidistant set of times.
In addition to trading the stock 
the investor allows to take static positions in vanilla options written on the
underlying asset at the maturity date. For the exponential utility maximization problem 
we establish in Theorem \ref{thm4.1} 
the continuous-time limit, i.e. the limit where the number of trading times goes to infinity.
The proof of this result is given in Section \ref{sec:5} and is based on Theorem 
\ref{thm1.1}. Next, in Theorem \ref{thm.100} we prove that 
the continuous-time hedging problem coincides with the scaling limit given by Theorem \ref{thm4.1}.
The lower bound for the continuous-time problem follows immediately from 
 the Fatou Lemma and Theorem \ref{thm4.1}. The proof of the upper bound is more technical and requires additional 
machinery which 
is based on duality theory and 
the theory of 
Radon-Nikodym derivatives of Gaussian measures (see \cite{H:68}). 
This is done in Section \ref{sec:6}.

We end this section with a discussion of our model. 
The realistic discrete-time setup is the following. The investor 
can trade dynamically the stock, in our case with delayed information. 
In addition she has a finite number of European call/put options
which can bought at the initial moment of time for a given (deterministic) price.
Our goal in this work is to provide an explicit solution for the exponential utility maximization problem, by applying various mathematical tools.
For this reason we consider several idealizations of the market model. 

The first idealization is that at time zero, the investor is able to buy any call option with strike 
$K\geq 0$ for the price 
$\int \left(x-K\right)^{+}d\hat\nu(x)$, where $\hat\nu$ is a given probability measure. 
The measure $\hat\nu$ is assumed to be derived from observed call prices that are liquidly traded in the market. Then, by a rough approximation argument we simplify the model
and assume that for any continuous function $f$ 
 we can buy an option with the random payoff $f(S_n)$ for the amount $\int f d\hat\nu$ (provided that the last term is well defined). 
Finally, in order to make the model tractable we assume that $\hat\nu$ is a normal distribution. 

It is natural to inquire how the normal distribution $\hat\nu$ 
is related to the (normal) distribution of the terminal stock price 
 with respect to the physical probability measure. In this paper we follow the no-arbitrage condition. 
 Consequently, in the discrete-time setup the only assumption we make is that the mean value of the distribution $\hat\nu$ is equal to the initial stock price. 
 In the continuous-time model, for strictly positive delay 
the same assumption is sufficient to guarantee no-arbitrage. For the continuous-time model with vanishing delay (i.e. $H=0$) 
we obtain the complete Bachelier model. To prevent arbitrage in this scenario,  
the distribution $\hat\nu$ must correspond to the unique martingale measure. More specifically, in order to avoid arbitrage for vanishing delay, 
the variance 
of the distribution $\hat\nu$ should be equal to the variance 
of the terminal asset price (see Remark \ref{rem.new}).

\section{The Discrete-Time Model}\label{sec:2}
Let $(\Omega,\mathcal F,\mathbb P)$ be a complete 
probability space carrying one risky asset which we denote by $S=(S_k)_{0\leq k\leq n}$ where 
$n\in\mathbb N$ is a fixed finite time horizon. We assume that the investor has a bank account that,
for simplicity, bears no interest. We assume that the initial stock price $S_0$ is a known constant 
and the stock increments 
$X_k:=S_k-S_{k-1}$, $k=1,...,n$ are i.i.d. normal variables with mean $\mu\in\mathbb R$ and 
variance equal to $\sigma^2$ for some $\sigma>0$. 
 
We fix a nonnegative integer number $\delay\in\mathbb Z_{+}$ and consider a situation where
there is a delay of $\delay$ trading times. Without loss of generality we assume that $\delay<n$.
 Hence, the investor's flow of information is given by the filtration 
$\mathcal G_k:=\sigma\left\{S_0,...,S_{\left(k-\delay\right)\vee 0}\right\}$, $k=0,1,...,n$. 
The case $\delay=0$ means that there is no delay and corresponds 
to the usual setting. 

We consider the following semistatic hedging setup.
There are two investment opportunities.
One is trading dynamically the stock at times $0,1,...,n-1$
 and the other is taking a static position in $S_n$, i.e. for any continuous function $f:\mathbb R\rightarrow\mathbb R$
 we can buy an option with the random payoff $f(S_n)$ for the amount $\int f d\hat\nu$ where $\hat\nu$ is a given measure (determined by the market).
 We assume that $\hat\nu$ is a normal distribution. 
Observe that 
the market allows to buy one stock for the price $S_0$ and so, in order to avoid an obvious arbitrage we should require 
$\int z d\hat\nu(z)=S_0$. Hence, we assume that $\hat\nu=\mathcal N(S_0, n\hat\sigma^2)$ for some $\hat\sigma>0$. 

Formally, a trading strategy is a pair $\pi=(\gamma,f)$ where $\gamma=(\gamma_1,...,\gamma_n)$ is a predictable process with respect
to $\{\mathcal G_k\}_{k=0}^n$ and $f:\mathbb R\rightarrow\mathbb R$ is a continuous function such that 
$\int |f| d\hat\nu<\infty$.
Denote by $\mathcal A$ the set of all trading strategies.
For $\pi=(\gamma,f)\in\mathcal A$ the corresponding portfolio value
at the maturity date is given by 
\begin{equation}\label{eq:value}
V^{\pi}_n=f(S_n)+\sum_{i=1}^n \gamma_i X_i-\int fd\hat\nu.
\end{equation}

\begin{rem}\label{rem2}
We say that two trading strategies $\pi=(\gamma,f)$ and $\tilde\pi=(\tilde\gamma,\tilde f)$ are equivalent if there exists a constant
$c\in\mathbb R$ such that $\gamma-\tilde\gamma\equiv c$ and 
$\tilde f-f$ is a linear function with slope equal to $c$. 
Clearly if $\pi=(\gamma,f)$ and $\tilde\pi=(\tilde\gamma,\tilde f)$ are equivalent
then $V^{\pi}_n=V^{\tilde\pi}_n$. Let us argue that the opposite direction also holds. Indeed, let 
$\pi=(\gamma,f)$ and $\tilde\pi=(\tilde\gamma,\tilde f)$ be two trading strategies 
such that $V^{\pi}_n=V^{\tilde\pi}_n$. Then, from (\ref{eq:value}) we obtain that the random variable 
$f(S_n)-\tilde f(S_n)+\left(\gamma_{n}-\tilde\gamma_{n}\right)S_n$ is 
 measurable with respect to $\sigma\{S_0,...,S_{n-1}\}$. Since the support of $S_n$ given 
 $\sigma\{S_0,...,S_{n-1}\}$ contains more than two values (it is all the real line) we conclude that 
 $f-\tilde f$ is a linear function and so $\pi=(\gamma,f)$ and $\tilde\pi=(\tilde\gamma,\tilde f)$ are equivalent. 
\end{rem}

Next, the investor’s preferences are described by an exponential utility function 
$u(x)=-\exp(- x)$, $x\in\mathbb R$
and her goal is to
\begin{equation}\label{1.2}
\mbox{Maximize} \ \ \mathbb E_{\mathbb P}\left[-\exp\left(- V^{\pi}_n\right)\right] \ \ \mbox{over} \ \ \pi\in\mathcal A
\end{equation}
where $\mathbb E_{\mathbb P}$ denotes the expectation with respect to the 
market probability measure $\mathbb P$. Clearly, 
for any portfolio strategy 
$\pi$ and a constant $\lambda\in\mathbb R$, 
$V^{\lambda\pi}_n=\lambda V^{\pi}_n$. 
Thus, without loss of generality, we take 
the absolute risk aversion parameter to be equal $1$. 
The following theorem is the cornerstone of the paper. 
\begin{thm}\label{thm1.1}
Fix $n,D,\mu,\sigma,\hat\sigma$. Let 
\begin{equation}\label{eq:a_sol}
a:=\begin{cases}
			\frac{\sigma^2}{2n\hat\sigma^2}+\frac{\sqrt{\left(2D+1-\frac{D(D+1)\sigma^2}{n\hat\sigma^2}\right)^2- 4 (D+1)D \left(1-\frac{\sigma^2}{\hat\sigma^2}\right)}-2D-1}{2 D (D+1)}, & \text{if \ $D\neq 0$}\\
 \frac{\sigma^2}{\hat\sigma^2}-1, & \text{if \ $D=0$}
		 \end{cases}
\end{equation}
and define the sequence $\{b_i\}_{i=1}^{\infty}$
by recursion ($D>0$)
\begin{equation}\label{1.4}
b_1=...=b_{D}=a, \ \ \mbox{and} \ \ b_i=\frac{a }{a D+1} \sum_{j=1}^{D} b_{i-j}\ \ \mbox{for} \ \ i>D,
\end{equation}
for $D=0$ we set $b_i=0$ for all $i>0$.

Then, $a>-\frac{1}{D+1}$ and so, the above sequence is well defined. 
The maximizer $\pi^{*}=(\gamma^{*},f^{*})$ for the optimization problem (\ref{1.2}) is unique (up to equivalency) and is given by 
\begin{equation}\label{1.5}
\gamma^{*}_i=\frac{\mu}{\sigma^2}+\frac{1}{\sigma^2}\sum_{j=1}^{i-1} (b_{i-j}-a) X_{j} \ \ \mbox{and} \ \ f^{*}(S_n):=\frac{a}{2\sigma^2}(S_n-S_0)^2.
\end{equation}
The corresponding value is given by 
\begin{align}\label{1.6}
u(n,D,\mu,\sigma,\hat\sigma) &:=\sup_{\pi\in\mathcal A}\mathbb E_{\mathbb P}\left[-\exp\left(- V^{\pi}_n\right)\right]=\mathbb E_{\mathbb P}\left[-\exp\left(- V^{\pi^*}_n\right)\right]\nonumber\\
&=-\exp\left(\frac{1}{2\sigma^2} n(a\hat\sigma^2-\mu^2)\right)
\sqrt{\frac{\big(1+ \delay a\big)^{n-\delay-1}}{\big(1+(\delay+1) a\big)^{n-\delay}}}.
\end{align}
\end{thm}
\vspace{10pt}

Let us briefly collect some financial-economic observations from this result. First,
Theorem \ref{thm1.1} states that the optimal trading strategy $\gamma^{*}$
is a sum of the well known
Merton fraction $\frac{\mu}{\sigma^2}$ and a centered discrete-time Gaussian process which we compute 
by recursive relations. 
The static option $f^{*}(S_T)$ is given by a quadratic payoff function.

Observe that for the case $\sigma=\hat\sigma$ we have $a=0$
and the optimal strategy is to hold $\frac{\mu}{\sigma^2}$ stocks (the well known Merton fraction) and not invest in static options. The corresponding value is 
$u(n,D,\mu,\sigma,\sigma)=-\exp\left(-\frac{n\mu^2}{2\sigma^2}\right).$ 

Next, we focus on the interesting case which is $\sigma\neq\hat\sigma$.
First, for $D=0$ (i.e. no delay) we have 
$a=\frac{\sigma^2}{\hat\sigma^2}-1$ and the value is given by 
$$u(n,0,\mu,\sigma,\hat\sigma)=-\exp\left(-\frac{n\mu^2}{2\sigma^2}\right)\exp\left(-\frac{n g\left(\frac{\hat\sigma^2}{\sigma^2}\right)}{2}\right)$$
where 
$g(z):=z-\log z-1$, $z>0$. 
The function $g\geq 0$ is convex and attains its minimal value (equal to zero) at $z=1$. This function 
also appears in the Brownian setup 
when dealing with a scaling limit of relative entropies associated with appropriate time discretizations 
of Brownian martingale and Brownian motion, for details see \cite{backhoff2023most} and the references therein.

For general $D$ and $\sigma\neq\hat\sigma$ we will see in Section \ref{sec:3} that $a$
is the largest (real) root of the quadratic equation (\ref{equat}). 
If $\hat\sigma>\sigma$ 
and $D>0$ then the sum of the roots of the quadratic equation 
(\ref{equat}) is 
$\frac{\frac{D(D+1)\sigma^2}{n\hat\sigma^2}-2D-1}{D(D+1)}<0$ and the product of the roots is 
$\frac{1-\frac{\sigma^2}{\hat\sigma^2}}{D(D+1)}>0$. Hence, in this case both of the roots are negative, and in particular $a<0$.
Namely, the investor short sells the quadratic option $(S_n-S_0)^2$. 
For the case $\sigma>\hat\sigma$ the product of the roots is 
$\frac{1-\frac{\sigma^2}{\hat\sigma^2}}{D(D+1)}<0$, and so the large root $a$ is positive. 

Next, fix $n,D,\mu,\sigma$. In view of the above we have $a>0\Leftrightarrow\sigma>\hat\sigma$. Hence, 
the value function 
$u(n,D,\mu,\sigma,\hat\sigma)$ as a function of $\hat\sigma$ is increasing for $\hat\sigma>\sigma$ and decreasing for $\hat\sigma<\sigma$. Indeed, let $\sigma<\hat\sigma_1<\hat\sigma_2$
and let $\pi^{*}_{\hat\sigma_1}$ be the optimal portfolio for $\hat\sigma=\hat\sigma_1$. Since the corresponding static option
is $\frac{a}{2\sigma^2}(S_n-S_0)^2$ with $a<0$, then the value (at the maturity date) of the same portfolio for the case $\hat\sigma=\hat\sigma_2$ will be larger.
Similarly, if $\hat\sigma_1<\hat\sigma_2<\sigma$ and $\pi^{*}_{\hat\sigma_2}$ is the optimal portfolio for $\hat\sigma=\hat\sigma_2$ then due to the fact that the corresponding 
$a$ is positive, the value of the same portfolio for the case $\hat\sigma=\hat\sigma_1$ will be larger.
Finally, we argue that there is an asymptotic arbitrage for the case where $\hat\sigma$ is very large or very small.
To this end, instead of applying Theorem \ref{thm1.1}, we provide a straightforward argument based on static hedging.
Indeed, by 
taking $\gamma\equiv 0$ and $f(S_n)=\frac{1}{\hat\sigma}(S_n-S_0)^2$ it follows that
\begin{equation}\label{static1}
\lim_{\hat\sigma\rightarrow 0} u(n,D,\mu,\sigma,\hat\sigma)=0 
\end{equation}
and by taking 
 $\gamma\equiv 0$ and $f(S_n)=-\frac{1}{\hat\sigma}(S_n-S_0)^2$
we obtain 
\begin{equation}\label{static2}
\lim_{\hat\sigma\rightarrow \infty} u(n,D,\mu,\sigma,\hat\sigma)=0.
\end{equation}

\section{Proof of Theorem \ref{thm1.1}}\label{sec:3}
Denote by $\mathcal Q$ the set of all equivalent probability measures $\mathbb Q\sim\mathbb P$ with finite entropy 
 $\mathbb E_{\mathbb Q}\left[\log\left(\frac{d\mathbb Q}{d\mathbb P}\right)\right]<\infty$ relative to $\mathbb P$ that satisfy 
 \begin{equation}\label{2.0}
 \mathbb E_{\mathbb Q}[S_j-S_i|\mathcal G_i]=0 \ \ \forall j\geq i
 \end{equation}
 and $(S_n;\mathbb Q) \sim \hat\nu$ (the latter means that the distribution of $S_n$ under $\mathbb Q$
 is equal to $\hat\nu$). We start with the following verification result. 
\begin{lem}\label{lem2.1}
If a triplet $(\tilde\pi,\tilde{\mathbb Q},C)\in \mathcal A\times \mathcal Q\times\mathbb R$ satisfies
\begin{equation}\label{2.1}
V^{\tilde\pi}_n+\log\left(\frac{d{\tilde{\mathbb Q}}}{d\mathbb P}\right)=C
\end{equation}
 then $\tilde\pi\in\mathcal A$ is the unique (up to equivalency) optimal portfolio for the optimization 
 problem (\ref{1.2}) and the corresponding value is 
 $
\mathbb E_{\mathbb P}\left[-\exp\left(- V^{\tilde\pi}_n\right)\right]=-e^{-C}.$
\end{lem}
\begin{proof}
The proof will be done in two steps. \\
${}$\\
 \textbf{Step I:}
In this step we prove that if $\pi\in\mathcal A$ satisfies 
$\mathbb E_{\mathbb P}\left[\exp\left(-V^{ \pi}_n\right)\right]<\infty$ then 
for any $\mathbb Q\in\mathcal Q$ 
\begin{equation}\label{2.2}
\mathbb E_{\mathbb Q}[V^{\pi}_n]=0.
\end{equation}
Indeed, let $\pi\in\mathcal A$ be as above and 
$\mathbb Q\in\mathcal Q$. In view of the relations 
$\int |f| d\hat\nu<\infty$ and $(S_n;\mathbb Q) \sim \hat\nu$ we have 
$f(S_n)\in L^1(\mathbb Q)$ and 
\begin{equation}\label{2.3}
\mathbb E_{\mathbb Q} [f(S_n)]=\int f d\hat\nu. 
\end{equation}
Next,
from the classical Legendre-Fenchel duality inequality
$xy \leq e^x + y(\log y-1)$ for $x=\left(V^{\pi}_n\right)_{-}$ (as usual we set $z_{-}:=-\min(0,z)$) 
and 
$y=\frac{d\mathbb Q}{d\mathbb P}$ we obtain 
\begin{equation}\label{2.4}
\mathbb E_{\mathbb Q}\left[\left(V^{\pi}_n\right)_{-}\right]\leq 
\mathbb E_{\mathbb P}\left[\exp\left(-V^{ \pi}_n\right)\right]+
\mathbb E_{\mathbb Q}\left[\log\left(\frac{d\mathbb Q}{d\mathbb P}\right)\right]
<\infty.
\end{equation}
From (\ref{eq:value}) and (\ref{2.3})--(\ref{2.4}) we get
$\mathbb E_{\mathbb Q}\left[\left(\sum_{i=1}^n \gamma_i X_i\right)_{-}\right]<\infty$.
Thus, using similar arguments as in the proof of Lemma 2.1 in \cite{BD:2021}
we obtain (in view of (\ref{2.0})) that 
 $\mathbb E_{\mathbb Q}\left[\sum_{i=1}^n \gamma_i X_i\right]=0$. 
This together with (\ref{2.3}) yields (\ref{2.2}) and completes the first step. 
\\
${}$\\
 \textbf{Step II:}
We proceed using similar arguments to those in Step II of the proof of Theorem $1.1$ in \cite{dolinsky2023exponential}.
From (\ref{2.1})--(\ref{2.2}) for $\pi:=\tilde\pi$ we obtain
\begin{equation}\label{2.5}
\log\left(\mathbb E_{\mathbb P}\left[\exp\left(-V^{\tilde \pi}_n\right)\right] \right)=-C=- \mathbb E_{\tilde{\mathbb Q}}\left[\log
\left(\frac{d\tilde{\mathbb Q}}{d\mathbb P}\right)\right].
\end{equation}
Next, it is well known that for a strictly concave utility maximization problem, the 
optimal portfolio value is unique (if exists) and so in view of Remark \ref{rem2} and 
(\ref{2.5}), in order to complete the proof of the Lemma it remains to show that in general we have the inequality
\begin{equation}\label{2.6}
\log\left(\mathbb E_{\mathbb P}\left[\exp\left(-V^{\pi}_n\right)\right] \right)\geq -\mathbb E_{\mathbb Q}\left[\log\left(\frac{d\mathbb Q}{d\mathbb P}\right)\right] \ \ \forall (\pi,\mathbb Q)\in\mathcal A\times\mathcal Q.
\end{equation}
Let us establish (\ref{2.6}).
We assume that $\mathbb E_{\mathbb P}\left[\exp\left(-V^{ \pi}_n\right)\right]<\infty$ (otherwise (\ref{2.6}) is trivial). Then
for any $z\geq 0$
\begin{align*}
\mathbb E_{\mathbb P}\left[\exp\left(-V^{ \pi}_n\right)\right]
&=\mathbb E_{\mathbb P}\left[\exp\left(-V^{ \pi}_n\right)+z \frac{d{\mathbb Q}}{d\mathbb P}V^{\pi}_n\right]\\
&\geq \mathbb E_{\mathbb P}\left[z\frac{d{\mathbb Q}}{d\mathbb P}\left(1-\log
\left(z\frac{d{\mathbb Q}}{d\mathbb P}\right)\right)\right]\\
&=z-z\log z-z\mathbb E_{\mathbb Q}\left[\log
\left(\frac{d{\mathbb Q}}{d\mathbb P}\right)\right].
\end{align*}
The first equality is due to \eqref{2.2}.
The inequality follows from the Legendre-Fenchel duality
inequality $xy \leq e^x+y(\log y-1)$ by setting $x=-V^{\pi}_n$ and 
$y=z\frac{d\mathbb Q}{d\mathbb P}$.
The last equality is straightforward. 
From simple calculus it follows that the concave function 
$z\rightarrow z-z\log z-z\mathbb E_{\mathbb Q}\left[\log
\left(\frac{d{\mathbb Q}}{d\mathbb P}\right)\right]$, $z>0$ attains its maximum at $z^{*}:=\exp\left(-\mathbb E_{\mathbb Q}\left[\log
\left(\frac{d{\mathbb Q}}{d\mathbb P}\right)\right]\right)$ and the corresponding maximal value 
is also $z^{*}$. This completes the proof of (\ref{2.6}).
\end{proof}
Next, we introduce some notations.
Let $\band_{\delay}$ be the set 
of all $n \times n$ positive definite matrices $Q$ that satisfy $Q_{ij}=0$ for $|i-j|>\delay$. Namely, 
$\band_{\delay}$ is the set of all banded positive definite matrices with lower bandwidth and upper bandwidth equal to $\delay$. 
Let $\toeplitz_n$ be the set of all $n \times n$ symmetric Toeplitz matrices. Recall that a Toeplitz matrix is a matrix in which each descending diagonal from left to right is constant. 
Let $C$ be a $n \times n$ real matrix and let $I,J \in \{1,..,n\}^k$ be two integer vectors of length $k \leq n$ denoting row and column indices, satisfying $I = (i_1,..,i_k), J=(j_1,..,j_k)$ 
with $1 \leq i_1 < i_2 < .. < i_k \leq n$, $1 \leq j_1 < j_2 < .. < j_k \leq n$. 
We denote by $C_J^I$ the minor specified by these indices, i.e. 
the determinant of the $k \times k$ sub-matrix determined by taking the elements $C_{ij}$ for $i=i_1,..i_k ; j=j_1,..j_k$. 
For any two natural numbers $l \geq k$ we denote the vector $(k,k+1,...,l)$ by $[k:l]$.
\begin{lem}\label{lem2.2}
Let $\{b_i\}_{i=1}^{\infty}$ be the sequence given by (\ref{1.4}). Set 
$b_0=a+1$ and define the matrix 
$A\in \toeplitz_n$ by 
$A_{ij}=b_{|i-j|}$. Then $A^{-1}\in \band_{\delay}$ and 
the determinant of $A$ is given by 
\begin{equation}\label{2.8-}
|A|=\frac{\big(1+(\delay+1) a\big)^{n-\delay}}{\big(1+ \delay a\big)^{n-\delay-1}} . 
\end{equation}
\end{lem}
\begin{proof}
For $D=0$ the statement is trivial. Assume that $D>0$. The proof will be done in three steps. 
\\
${}$\\
 \textbf{Step I:}
In this step we prove that $a$, given by (\ref{eq:a_sol}), satisfies $a>-\frac{1}{D+1}$.

Observe that (recall that $D<n$) 
 $$
\lim_{z\rightarrow\infty}\frac{n+z D (D+1)}{\left(Dz+1\right)\left((D+1)z+1\right)}=0
$$ and
$$
\lim_{z\downarrow-\frac{1}{D+1}}\frac{n+z D (D+1)}{\left(Dz+1\right)\left((D+1)z+1\right)}=\infty
$$
 and so, the equation 
 $\frac{n+z D (D+1)}{\left(Dz+1\right)\left((D+1)z+1\right)}=\frac{\hat\sigma^2}{\sigma^2}n$ 
 has at least one real root $z> -\frac{1}{D+1}$. Hence, 
 the quadratic equation
 \begin{equation}\label{equat}
 D(D+1) z^2+\left(2D+1-\frac{D(D+1)\sigma^2}{n\hat\sigma^2}\right)z+1-\frac{\sigma^2}{\hat\sigma^2}=0
 \end{equation}
 has a real root $z> -\frac{1}{D+1}$. 
 Since $a$ is the largest root of the above quadratic equation it follows that 
 $a\geq z> -\frac{1}{D+1}$.
 \\
${}$\\
 \textbf{Step II:}
 In view of the first step, (\ref{1.4}) provides a well defined sequence of real numbers. Hence, 
 $A\in \toeplitz_n$. 
 In this step we show that the matrix $A$ has vanishing sub-$(D+1)$-minors. That is, all minors of the form 
 $A_J^I$ where $I = (i_1,..,i_{D+1}), J=(j_1,..,j_{D+1})$ 
 and $i_i>j_{\delay+1}-\delay$ are equal to zero.
 
 Set $b_{-i}:=b_i=a$ for $i=1,..,\delay$. Observe that 
the sequence $\{b_i\}_{i=-D}^{\infty}$ satisfies the recursion 
\begin{equation}\label{2.8}
b_i=\frac{a }{a D+1} \sum_{j=1}^{D} b_{i-j}\ \ \mbox{for} \ \ i>0.
\end{equation}
Let $A^J_I$ be a sub-$(D+1)$-minor. From the fact that $i_1>j_{D+1}-D$ we obtain that $A_{i_k j_l}=b_{|i_k-j_l|}=b_{i_k-j_l}$. This, together with (\ref{2.8}), yield that all the rows of the matrix that corresponds to the minor $A^J_I$ are determined as linear combinations of rows $i_1,i_1+1,..,i_1+\delay-1$ of $A$, taken at the appropriate columns $j_1,..,j_{\delay+1}$. Hence, the rank of this matrix is at most $\delay$ and $A_J^I=0$ as required. 
\\
${}$\\
\textbf{Step III:}
In this step we complete the proof. 
First, 
for any $m\in\mathbb N$ and $\alpha>-1/m$ define the $m \times m$ matrix $\mathbf {1}^{m,\alpha}$ by 
$\mathbf {1}^{m,\alpha}_{ij}:=\alpha+\mathbb I_{i=j}$ (as usual $\mathbb I$ denotes the indicator function). 
Then, for any non-zero vector $x \in \mathbb{R}^m$, $x^T \mathbf {1}^{m,\alpha} x = (1+\alpha m) ||x||^2 - \alpha \sum_{i<j} (x_i - x_j)^2 > 0$, hence $\mathbf {1}^{m,\alpha}$ is a positive definite matrix. In addition, the determinant is equal to 
$
|\mathbf {1}^{m,\alpha}|=1+m\alpha.
$
 Thus, in view of the second step, 
from Theorem 5.1 in \cite{BF:81} we obtain that 
$$
|A|=\frac{\prod\limits_{i=1}^{n-\delay} A_{[i:i+\delay]}^{[i:i+\delay]}}{\prod\limits_{i=1}^{n-\delay-1} A_{[i+1:i+\delay]}^{[i+1:i+\delay]}} = \frac{\big(1+(\delay+1) a\big)^{n-\delay}}{\big(1+ \delay a\big)^{n-\delay-1}}.
$$
 Moreover, from Theorem 3.1 in \cite{BF:81}, it follows that $A^{-1}$ is $\delay$-banded. Finally, from the fact that $a>-\frac{1}{D+1}$ we get that 
that the principal minors
$A^{[k+1:k+\delay]}_{[k+1:k+\delay]}>0$ for all $k\leq n-\delay-1$ and 
$A^{[k:k+\delay]}_{[k:k+\delay]}>0$ 
for all $k\leq n-\delay$.
Hence, by applying Theorem 5.5 in \cite{BF:81} for $r=s=\delay+1$ we obtain
that 
$A^{-1}\in\band_{\delay}$. 
\end{proof}

So far we did not use the specific definition of $a$.
This is done in the next computational result. 
 \begin{lem}\label{lem2.3}
 We have $\sum_{i,j=1}^n [A^{-1}]_{ij}=n\frac{\hat\sigma^2}{\sigma^2}$.
 \end{lem}
 \begin{proof}
 For $D=0$ we have $A^{-1}=\frac{\hat\sigma^2}{\sigma^2} I$ and so, the statement is trivial. Assume that $D>0$.
 The proof will be done in two steps. 
 \\
${}$\\
 \textbf{Step I:} 
 Set $v=(v_0,...,v_{n-1})$ by 
 \begin{equation}\label{2.14+}
 v_0=\frac{a D+1}{a(D+1)+1}, \ \ v_1=...=v_{D}=\frac{-a}{a(D+1)+1}, \ \ v_{D+1}=...=v_{n-1}=0.
 \end{equation}
 In this step we prove that 
 \begin{equation}\label{2.15}
\sum_{j=0}^{n-1} v_j b_{|i-j|}=\delta_{i0}, \ \ i=0,1...,n-1.
 \end{equation} 
 This is shown by straightforward derivation for the different cases of $i$. \\
For $i = 0$:
$$
\sum_{j=0}^{n-1} v_j b_{j}=\frac{(1+a)(a\delay+1)}{a (\delay+1)+1} - \frac{\delay a^2}{a (\delay+1)+1} = 1.
$$

\noindent For $1 \leq i < \delay$:
\begin{align*}
\sum_{j=0}^{n-1} v_j b_{|i-j|} &=
v_0 b_{i}+v_i b_0 +\sum_{j=1}^{i-1} v_j b_{i-j}+\sum_{j=i+1}^{\delay} v_j b_{j-i} \\
&= \frac{a (a D+1)- a(a+1)-a^2 (D-1)}{a(\delay+1)+1} = 0.
\end{align*}

\noindent For $i\geq \delay$: 
\begin{align*}
\sum_{j=0}^{n-1} v_j b_{|i-j|} &= \sum_{j=0}^{D} v_j b_{i-j}\\
&= \frac{a D+1}{a(\delay+1)+1} \left( b_{i}-\frac{a}{aD+1} \sum_{j=1}^{\delay} b_{i-j}\right)=0
\end{align*}
where the last equality follows from the recursive relation in (\ref{1.4}). This completes the first step. 
\\
${}$\\
\textbf{Step II:}
 In this step we complete the proof. 
 From Theorem 1
 in \cite{rodman1992inversion} and \eqref{2.15}
 \begin{equation}\label{2.16}
 [A^{-1}]_{ij}=\frac{1}{v_0}\left(\sum_{k=1}^{i\wedge j} v_{i-k}v_{j-k}-\sum_{k=1}^{i\wedge j-1 }v_{n-i+k}v_{n-j+k}\right).
 \end{equation}
 Thus,
 \begin{align*} 
 \sum_{i,j=1}^{n} [A^{-1}]_{ij}
 &=\sum_{i=1}^n [A^{-1}]_{ii}+
 2\sum_{k=1}^ D\sum_{j=1}^{n-k} [A^{-1}]_{j+k j}\\
 &=n v_0+\frac{v^2_1}{v_0}\left(\sum_{i=1}^n \left((i-1)\wedge D-0\vee (D+i-n)\right)\right)\\
 &+2 v_1\sum_{k=1}^D (n-k)\\
 &+2 \frac{v^2_1}{v_0}\sum_{k=1}^D \sum_{j=1}^{n-k} \left((j-1)\wedge (D-k)-0\vee (D+j-n)\right)\\
 &= n v_0+(D-2n-D-1)v_1+D^2 (n-D-1)\frac{v^2_1}{v_0}\\
 &=\frac{(a D+1)n}{aD+1+a}-\frac{a D (2n-D-1)}{aD+1+a}+\frac{a^2 D^2 (n-D-1)}{(aD+1)(aD+1+a)}\\
 &=n\sigma^2. 
 \end{align*}
 The first equality is due to the fact that $A^{-1}\in\band_D$. The second equality follows from 
 (\ref{2.16}) and the fact that $v_1=...=v_D$ and $v_{D+1}=...=v_{n-1}=0$. The third equality is simple summation. 
The fourth equality is a substitution of $v_0,v_1$ from (\ref{2.14+}).
The last equality is due to the fact that $a>-\frac{1}{D+1}$ is a solution to 
(\ref{equat}).
This completes the proof. 
\end{proof} 

We now have all the pieces in place that we need for the completion of the proof of Theorem \ref{thm1.1}.
\begin{proof}
From Lemma \ref{lem2.2} it follows that the portfolio $\pi^{*}=(\gamma^{*},f^{*})$ which given by (\ref{1.5}) is well defined 
and, since $b_{i-j}-a=0$ for $i-D-1<j<i$ we have $\pi^{*}\in\mathcal A$.

Set
\be\label{2.8+} 
\hat C := \frac{1}{2\sigma^2}n\left(\mu^2-a{\hat\sigma^2}\right) + \frac{1}{2} \log \left( \frac{\left(1+\left(D+1\right)a\right)^{n-\delay}}{\left(1+Da\right)^{n-\delay-1}} \right) \\ 
%&+\frac{1}{2}\left(\left(n-D\right)\log \left(1+\left(D+1\right)a\right)-\left(n-D-1\right)\log\left(1+Da\right)\right)\nonumber
\ee 
and define the measure 
$\hat{\mathbb Q}$ by 
\begin{equation}\label{2.9}
\frac{d\hat{\mathbb Q}}{d\mathbb P}:=\exp\left(\hat C-V^{\pi^{*}}_n\right).
\end{equation}
Clearly, the triplet 
$(\pi^{*},\hat{\mathbb Q},\hat C)$ satisfies (\ref{2.1}). 
Thus, in view of Lemma \ref{lem2.1} in order to complete the proof of Theorem \ref{thm1.1}
we need to show that 
$\hat{\mathbb Q}\in\mathcal Q$.

Let $X=(S_1-S_0,S_2-S_1,...,S_n-S_{n-1})$ and $\boldsymbol{\mu_n} =(\mu,...,\mu)\in\mathbb R^n$.
From (\ref{eq:value}), (\ref{1.5}) and the simple equality $(S_n-S_0)^2=\sum_{i=1}^n X^2_i+2\sum_{1\leq i<j\leq n} X_iX_j$
we have
 $$V^{\pi^{*} }_n=\frac{1}{2\sigma^2}\left(X (A-I) X{'}+2 \boldsymbol{\mu_n} X'-n a \hat\sigma^2\right).$$
 This together with (\ref{2.8-}) and (\ref{2.8+})--(\ref{2.9}) yield 
 $$\frac{d\hat{\mathbb Q}}{d\mathbb P}:=\frac{\frac{\exp\left(-\frac{1}{2}X \frac{A}{\sigma^2} X'\right)}{\sqrt {(2\pi)^n | A^{-1}|}}}{\frac{\exp\left(-\frac{1}{2}(X-{\boldsymbol{\mu_n}})\frac{I}{\sigma^2}(X-{\boldsymbol{\mu_n}})'\right)}{\sqrt {(2\pi)^n}}}.$$
Hence, the relation $\left(X;{\mathbb P}\right) \sim\mathcal N(\boldsymbol{\mu_n},\sigma^2 I)$
implies that $\hat{\mathbb Q}$ is a probability measure 
and $\left(X;\hat{\mathbb Q}\right) \sim\mathcal N(0,\sigma^2 A^{-1})$.
In particular $\mathbb E_{\hat{\mathbb Q}}\left[ \log\left(\frac{d\hat{\mathbb Q}}{d\mathbb P}\right)\right]<\infty$. 
From Lemma \ref{lem2.2} we obtain that $A^{-1}\in\mathcal S_D$ and so
(under the probability measure $\hat{\mathbb Q}$),
 $X_k$ is independent of $(X_1,...,X_{(k-1-\delay)^{+}})$ for all $k$. This yields
 (\ref{2.0}). 
 Finally, the property $(S_n;\hat{\mathbb Q}) \sim \hat\nu$
 follows from Lemma \ref{lem2.3}.
 \end{proof}

\section{The Bachelier Model}\label{sec:4}
In this section we assume that $(\Omega, \mathcal{F}, \mathbb P)$ is
 a complete
probability space carrying a one-dimensional Brownian motion 
$W=(W_t)_{t \in [0,1]}$ with time horizon $T=1$. 
We consider a simple financial market with a riskless savings account bearing
zero interest (for simplicity) and with a risky asset $P$ with Bachelier price
dynamics
${P}_t={P}_0+\theta t+\varsigma W_t, \ \ t\in [0,1]$
 where $P_0\in\mathbb R$ is the initial asset price, $\theta\in\mathbb R$ is the constant
drift and $\varsigma>0$ is the constant volatility. 
We focus on an investor who is informed about the risky asset’s price changes
with a delay $H>0$. 

\subsection{Scaling Limit}\label{section4.1}
Fix $n\in\mathbb N$ and consider an investor who can trade and observe the risky asset only at times from the grid
$\left\{0,\frac{1}{n},\frac{2}{n},...,1\right\}$. 
The investor's flow of information is given by the filtration 
$\mathcal G^n_k:=\sigma\left\{P_0,P_{\frac{1}{n}},...,P_{\frac{(k-D_n)\vee 0}{n}}\right\}$, $k=0,1,...,n$ 
where $D_n:=\min\{m\in\mathbb N: m /n\geq H\}=\lceil H n\rceil$.
In addition, the investor can take a static position in $P_1$, i.e. for any continuous function $f:\mathbb R\rightarrow\mathbb R$
 we can buy a vanilla option with the random payoff $f(P_1)$ for the amount $\int f d\hat{\mathcal{V}}$. We assume 
 that $\hat{\mathcal{V}}= \mathcal N(P_0, \hat\varsigma^2)$ 
for some constant $\hat\varsigma>0$. 
As in Section \ref{sec:2}, in order to avoid arbitrage the mean of the distribution 
$\hat{\mathcal{V}}$ should be equal to the initial stock price $P_0$. 

A trading strategy is a pair $\pi=(\gamma,f)$ where $\gamma=(\gamma_t)_{t \in [0,1]}$ is a piecewise constant process 
that satisfies 
\begin{equation}\label{4.0-}
\gamma_t:=\begin{cases}
			\gamma_{\frac{k}{n}}, & \text{if \ $\frac{k}{n} \leq t <\frac{k+1}{n}$, \ \ $k=0,1,...,n-2$}\\
 \gamma_{\frac{n-1}{n}}, & \text{if \ $\frac{n-1}{n}\leq t\leq 1$}
		 \end{cases}
\end{equation}
and $f:\mathbb R\rightarrow\mathbb R$ is a continuous function such that 
$\int f d\hat{\mathcal{V}}<\infty$.
Moreover, $\gamma_{\frac{k}{n}}$ is $\mathcal G^n_k$ measurable for all $k=0,1,...,n-1$.
Denote by $\mathcal A^n$ the set of all trading strategies.
For $\pi=(\gamma,f)\in\mathcal A^n$ the corresponding portfolio value
at the maturity date is given by 
$$
V^{\pi}_1=f(P_1)+\int_{0}^1 \gamma_t dP_t -\int fd\hat{\mathcal{V}}.
$$
Observe that for $\gamma$ of the form (\ref{4.0-}) we have 
$\int_{0}^1 \gamma_t dP_t=\sum_{i=1}^n \gamma_{\frac{i-1}{n}} \left(P_{\frac{i}{n}}-P_{\frac{i-1}{n}}\right).$

Consider the exponential utility maximization problem
\begin{equation}\label{4.1}
U_n(H,\theta,\varsigma,\hat\varsigma):=\sup_{\pi\in\mathcal A^n}\mathbb E_{\mathbb P}\left[-\exp\left(- V^{\pi}_1\right)\right].
\end{equation}

Before we formulate the limit theorem for the above problem, we will need some preparations. 
 Set 
\begin{equation}\label{eq:alpha_def}
 \alpha:=\alpha(H,\varsigma,\hat\varsigma)=
 \frac{1}{H}\left(1-\frac{2}{H\frac{\varsigma^2}{\hat\varsigma^2}+\sqrt{4(1-H)\frac{\varsigma^2}{\hat\varsigma^2}+H^2\frac{\varsigma^4}{\hat\varsigma^4}}}\right).
 \end{equation}
 Observe that $\alpha$ is determined by $H$ and the ratio $\frac{\varsigma}{\hat\varsigma}$.

Define the function $\kappa:[0,1]\rightarrow\mathbb R$ by 
\begin{equation}
\kappa_t = \frac{\alpha}{1-\alpha H} + \sum_{k=1}^{\lceil \frac{1}{H} \rceil - 1} {\mathbb I}_{t \in [kH, (k+1)H)} \exp\big(\alpha\left(t- k H\right)\big) \sum_{j=0}^{k-1} c_{k-j} \frac{(-\alpha)^j}{j!} (t-kH)^{j}
\label{eq:b_sol_recursive}
\end{equation} 
where $c_{1}=-\alpha$ and the remaining constants $c_{k+1}$ for $k=1,..,\lceil \frac{1}{H}-1 \rceil $ are given by the recurrence relation
\be
c_{k+1} =\exp\left(\alpha H\right)\sum_{j=0}^{k-1}c_{k-j}\frac{(-\alpha H)^j}{j!} .
\label{eq:c_coefficients_recursive}
\ee
From (\ref{eq:alpha_def}) we have $\alpha<\frac{1}{H}$ and so $\kappa$ is well defined.

We next explain the intuition behind the definition of $\kappa$. The fact that $c_1=-\alpha$ gives 
$\kappa_H=\frac{\alpha^2 H}{1-\alpha H}=\int_{0}^H \kappa_s ds$. 
The recursive definition (\ref{eq:c_coefficients_recursive}) provides that 
$\kappa$ is continuous in the points $2H,3H,...$ and so $\kappa$ is continuous on the interval $[H,1]$. 
From 
(\ref{eq:b_sol_recursive}) and the product rule (for differentiation) we obtain that for any $t\in (kH,(k+1)H)$ 
\begin{align*}
\frac{d}{dt}\kappa_t-\alpha \kappa_t &=-\frac{\alpha^2}{1-\alpha H}
+\exp\big(\alpha\left(t- k H\right)\big) \sum_{j=1}^{k-1} j c_{k-j} \frac{(-\alpha)^j}{j!} (t-kH)^{j-1}\\
&=-\frac{\alpha^2}{1-\alpha H}-\alpha
\exp\big(\alpha\left(t- k H\right)\big) \sum_{j=0}^{k-2} c_{k-1-j} \frac{(-\alpha)^j}{j!} (t-kH)^{j}\\
&=-\alpha\kappa_{t-H}.
\end{align*}
We conclude that $\kappa$ satisfies the following (integral) equation
\begin{equation}\label{prop}
\kappa_t=\alpha \int_{t-H}^t \kappa_s ds, \ \ \forall t\geq H. 
\end{equation} 
In fact our definition of $\kappa$ 
by (\ref{eq:b_sol_recursive})-(\ref{eq:c_coefficients_recursive})
was done after solving explicitly the linear equation (\ref{prop}) with the initial condition 
$\kappa_t=\frac{\alpha}{1-\alpha H}$ for $t<H$. The solution is obtained recursively. For any $k\in\mathbb N$,
$\kappa_{(kH,(k+1)H}$ solves the simple ordinary differential equation 
$\frac{d}{dt}\kappa_t-\alpha\kappa_t=-\alpha\kappa_{t-H}$ and $\kappa_{t-H}$ is known from previous step. We omit the technical details for this part.
The only details needed are
the definition of $\kappa$ 
by (\ref{eq:b_sol_recursive})-(\ref{eq:c_coefficients_recursive})
and the property (\ref{prop}).

We arrive at the limit theorem.
\begin{thm}\label{thm4.1}
As $n\rightarrow\infty$ we 
have 
\begin{align}\label{4.0}
&\lim_{n\rightarrow\infty} U_n(H,\theta,\varsigma,\hat\varsigma)\nonumber\\
&=
-{\exp\left(\frac{1}{2}\left(-\frac{\theta^2}{\varsigma^2}+
\alpha\left(\frac{\hat\varsigma^2}
{\varsigma^2(1-\alpha H)}+H-1\right)\right)\right)}\sqrt{1-\alpha H}.
\end{align}
 Moreover, for any $n\in\mathbb N$ let 
 $\pi^{*,n}=(\gamma^{*,n},f^{*,n},)\in\mathcal A^n$ be the optimal portfolio (up to equivalence) for the optimization problem (\ref{4.1}). Then,
 \begin{equation}\label{4.2}
 \lim_{n\rightarrow\infty} f^{*,n}(P_1)=\frac{\alpha}{2\varsigma^2(1-\alpha H)}\left(P_1-P_0\right)^2, 
 \end{equation}
and we have the uniform convergence 
 \begin{equation}\label{4.3}
 \lim_{n\rightarrow\infty}\sup_{0\leq t\leq T}|\gamma^{*,n}_t-\Gamma_t|=0 \ \ \mbox{a.s.}
 \end{equation}
 where $\Gamma$ is the Volterra Gaussian process given by
\begin{equation}\label{vol}
\Gamma_t:=\frac{\theta}{\varsigma^2}+\frac{1}{2\varsigma^2}\int_{0}^t \left(\kappa_{t-s}-\frac{\alpha}{1-\alpha H}\right)dP_s.
\end{equation}
\end{thm}

The proof of Theorem \ref{thm4.1} is given in Section \ref{sec:5}. 
Specifically, in order to prove the uniform convergence of the piecewise constant process ${\gamma}^{*,n}$ to the continuous Voltera Gaussian Process $\Gamma$, recall that for the discrete 
case ${\gamma}^{*}_i$ from \eqref{1.5} 
is determined by the coefficients $b_i$ defined in \eqref{eq:a_sol}--\eqref{1.4}, 
whereas $\Gamma$ in \eqref{vol} is determined by the continuous 
function $\kappa$ defined in \eqref{eq:b_sol_recursive}--\eqref{eq:c_coefficients_recursive}. 
In Lemma \ref{lem5.2} we show that the coefficients $b_i$ converge when 
properly scaled by $n$ to the continuous weight function $\kappa$ as $n \to \infty$. This convergence is illustrated numerically in Figure \ref{fig:strategy_limit_prop} (bottom) in Section \ref{sec4.2}. See details about the computations of the optimal strategy and value in Appendix \ref{sec:appendix}.

\begin{figure}[!ht]
\centering
\includegraphics[width=0.8\columnwidth]{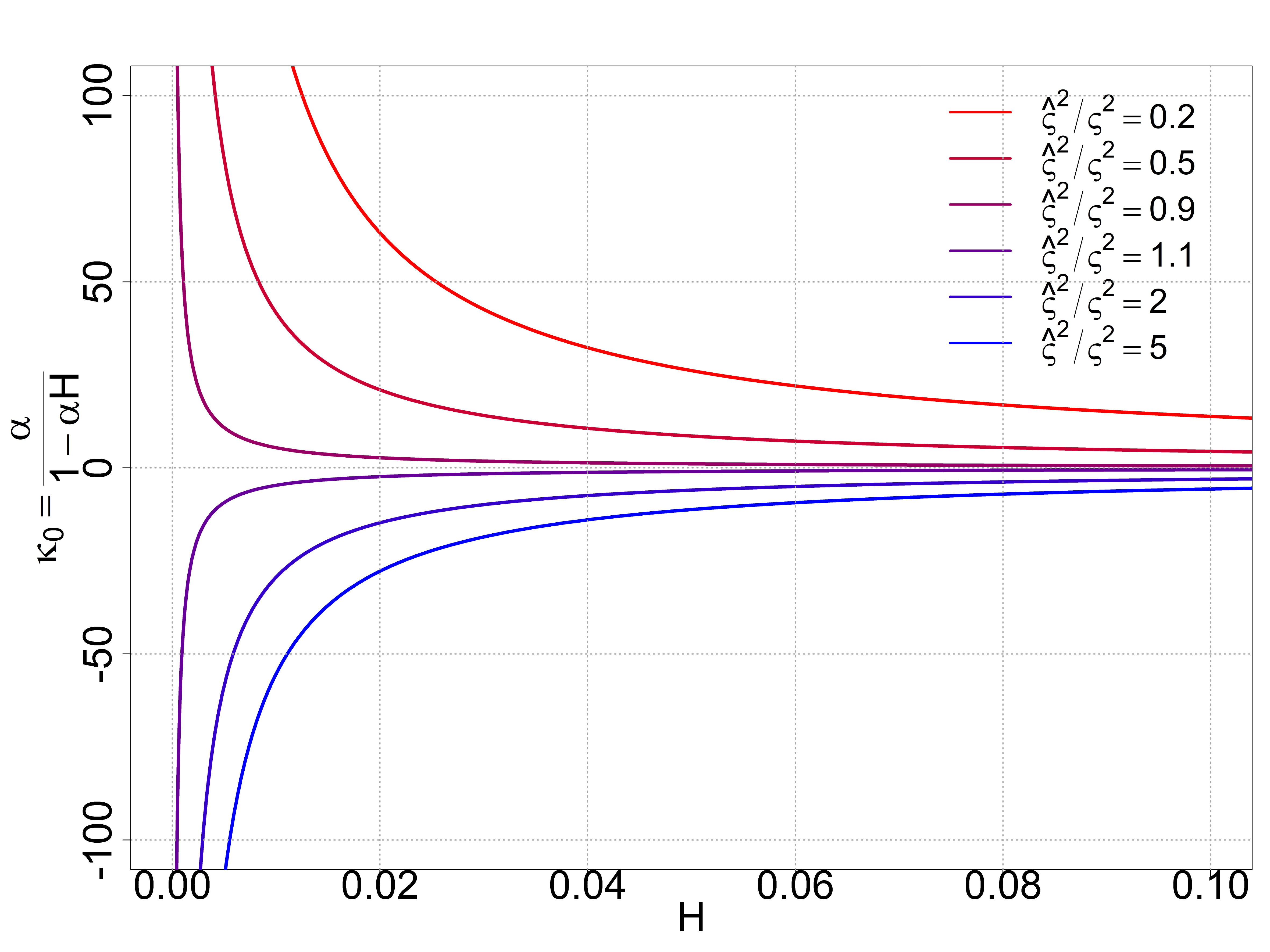}
\caption{\footnotesize 
The value of $\kappa_0 = \frac{\alpha}{1-\alpha H}$ for different values of the ratio $\frac{\hat \varsigma^2}{\varsigma^2}$. 
The initial investment weight in the static option $\kappa_0$ is positive (negative) for $\frac{\hat \varsigma^2}{\varsigma^2} < 1$ ($>1$). When the delay is small ($H \downarrow 0$), by \eqref{eq:alpha_def} the investment weight diverges, $\kappa_0 \sim \frac{1}{H} (1-\frac{\hat \varsigma^2}{\varsigma^2} ) \to \pm \infty$ (provided that $\hat \varsigma\neq\varsigma$).
 The value of $|\kappa_0|$ decreases monotonically with $H$.
\label{fig:kappa_0_vs_H}}	
\end{figure}

\subsection{Continuous-Time Hedging}\label{sec4.2}
In this section we consider the setup where trading the risky asset is done continuously.
Let $\{\mathcal G_t\}_{t=0}^1$ be the augmented filtration which is generated 
by the process $\{P_{(t-H)\vee 0}\}_{t=0}^1$. 
A trading strategy is a 
pair $\pi=(\gamma,f)$ where $\gamma=\left(\gamma_t\right)_{t \in [0,1]}$ is a 
predictable process with respect 
to $(\mathcal G_t)_{t\in [0,1]}$ which satisfies $\int_{0}^1 \gamma^2_t dt<\infty$ a.s. 
and $f:\mathbb R\rightarrow\mathbb R $ is a continuous function such that 
$\int f d\hat{\mathcal V}<\infty$ (the probability measure $\hat{\mathcal V}$
is the same as in Section \ref{section4.1}).
For $\pi=(\gamma,f)$ the corresponding portfolio value
at the maturity date is given by 
\begin{equation}\label{9.1}
V^{\pi}_1=f(P_1)+\int_{0}^1 \gamma_t dP_t-\int f d\hat{\mathcal V}
\end{equation}
where the integral $\int_{0}^1 \gamma_t dP_t$ is the It\^{o} integral.
Let us emphasize that we do not require any notion of admissibility for
the portfolio value (we do assume that the constant delay $H$ is strictly larger than zero).
Denote by $\mathfrak A$ the set of all trading strategies. 

\begin{rem}\label{rem20}
Similarly to Section \ref{sec:2} 
we say that two trading strategies $\pi=(\gamma,f)$ and $\tilde\pi=(\tilde\gamma,\tilde f)$ are equivalent if there exists a constant
$c\in\mathbb R$ such that $\gamma-\tilde\gamma=c$ $dt\otimes\mathbb P$ a.s and 
$\tilde f-f$ is a linear function with slope equal to $c$. 
Clearly if $\pi=(\gamma,f)$ and $\tilde\pi=(\tilde\gamma,\tilde f)$ are equivalent
then $V^{\pi}_1=V^{\tilde\pi}_1$ a.s.
Let us argue that the opposite direction also holds. Indeed, let 
$\pi=(\gamma,f)$ and $\tilde\pi=(\tilde\gamma,\tilde f)$ be two trading strategies 
such that $V^{\pi}_1=V^{\tilde\pi}_1$.
Set $\beta=\gamma-\tilde\gamma$ and $h=\tilde f-f$. 
Let $n\in\mathbb N$ such that $\frac{1}{n}<H$.
Introduce the process 
$B_t:=W_{t+1 -\frac{1}{n}}-W_{1-\frac{1}{n}}$, $t\in \left[0,\frac{1}{n}\right]$. From $V^{\pi}_1=V^{\tilde\pi}_1$ we get
\begin{equation}\label{100}
h\left(W_{1-\frac{1}{n}}+B_{\frac{1}{n}}\right)=\int_{0}^{1-\frac{1}{n}}\beta_t dW_t+\int_{0}^{\frac{1}{n}}\beta_{s+1-\frac{1}{n}}dB_s.
\end{equation}
Observe that
$B$ is a Brownian motion which is independent of $\left(W_t\right)_{t\in\left[0,1-\frac{1}{n}\right]}$ and
$(\beta_t)_{t\in [0,1]}$ (here we use that $\frac{1}{n}<H$). Hence, 
(\ref{100}) implies that there exists a constant $c\in\mathbb R$ such that $h$ is a linear
function with slope equal to $c$ and the restriction of the process $\beta$ to the interval
$\left[1-\frac{1}{n},1\right]$ is equal $c$ ($dt\otimes \mathbb P$ a.s.).
By repeating this argument for the intervals $\left[\frac{k-1}{n},\frac{k}{n}\right]$, $k=n-1,n-2,...,1$
 we conclude that $\beta=c$ $dt\otimes\mathbb P$ as required. 
\end{rem}

The following theorem provides the solution to the continuous-time hedging problem.
\begin{thm}\label{thm.100}
Let $\pi_{*}:=(\Gamma,F)$
where 
$\Gamma$ is given by 
(\ref{vol}) and $F(P_1)$ is equal to the 
right-hand side of (\ref{4.2}). Then, $\pi_{*}$ is the unique (up to equivalency) solution for the optimization problem:
$\mbox{Maximize} \ \mathbb E_{\mathbb P}\left[-\exp\left(- V^{\pi}_1\right)\right] \ \mbox{over} \ \pi\in\mathfrak A.
$
Moreover,
\begin{align*}
&\sup_{\pi\in\mathfrak A}\mathbb E_{\mathbb P}\left[-\exp\left(-V^{\pi}_1\right)\right]=
\mathbb E_{\mathbb P}\left[-\exp\left(-V^{\pi_{*}}_1\right)\right]\\
&=-{\exp\left(\frac{1}{2}\left(-\frac{\theta^2}{\varsigma^2}+
\alpha\left(\frac{\hat\varsigma^2}
{\varsigma^2(1-\alpha H)}+H-1\right)\right)\right)}\sqrt{1-\alpha H}.
\end{align*}
\end{thm}
\begin{proof}
Since the utility function is strictly concave then
the optimal portfolio value is unique (if exists). This together with Remark \ref{rem20} gives uniqueness. 

Next, recall the optimal portfolios $\pi^{*,n}$, $n\in\mathbb N$ from Theorem \ref{thm4.1}.
From (\ref{4.3}) it follows that $V^{\pi^{*,n}}\rightarrow V^{\pi_{*}}$ in probability (as $n\rightarrow\infty$). Thus, by 
applying the 
Fatou Lemma for the function $x\rightarrow e^{-x}$, we obtain from (\ref{4.0}).
$$\mathbb E_{\mathbb P}\left[-\exp\left(-V^{\pi}_1\right)\right]\geq -{\exp\left(\frac{1}{2}\left(-\frac{\theta^2}{\varsigma^2}+
\alpha\left(\frac{\hat\varsigma^2}
{\varsigma^2(1-\alpha H)}+H-1\right)\right)\right)}\sqrt{1-\alpha H}.$$

We conclude that in order to complete the proof of Theorem \ref{thm.100} it remains to establish the upper bound
\begin{equation}\label{9.2}
\sup_{\pi\in\mathfrak A}\mathbb E_{\mathbb P}\left[-\exp\left(-V^{\pi}_1\right)\right]
\leq -{\exp\left(\frac{1}{2}\left(-\frac{\theta^2}{\varsigma^2}+
\alpha\left(\frac{\hat\varsigma^2}
{\varsigma^2(1-\alpha H)}+H-1\right)\right)\right)}\sqrt{1-\alpha H}.
\end{equation}
This is done in Section \ref{sec:6}.  
\end{proof}

\begin{figure}[!ht]
\includegraphics[width=0.495\columnwidth]{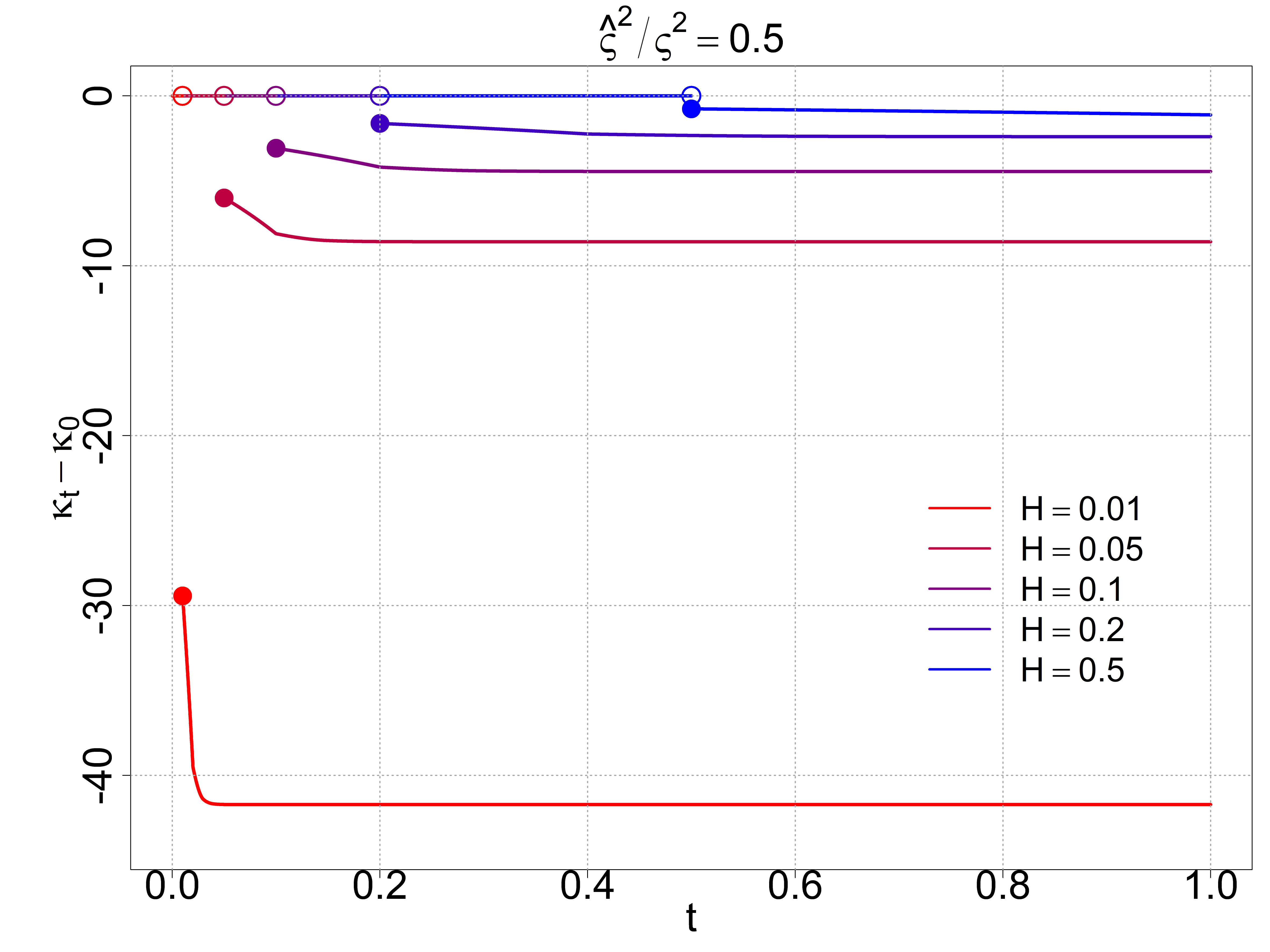}
\includegraphics[width=0.495\columnwidth]{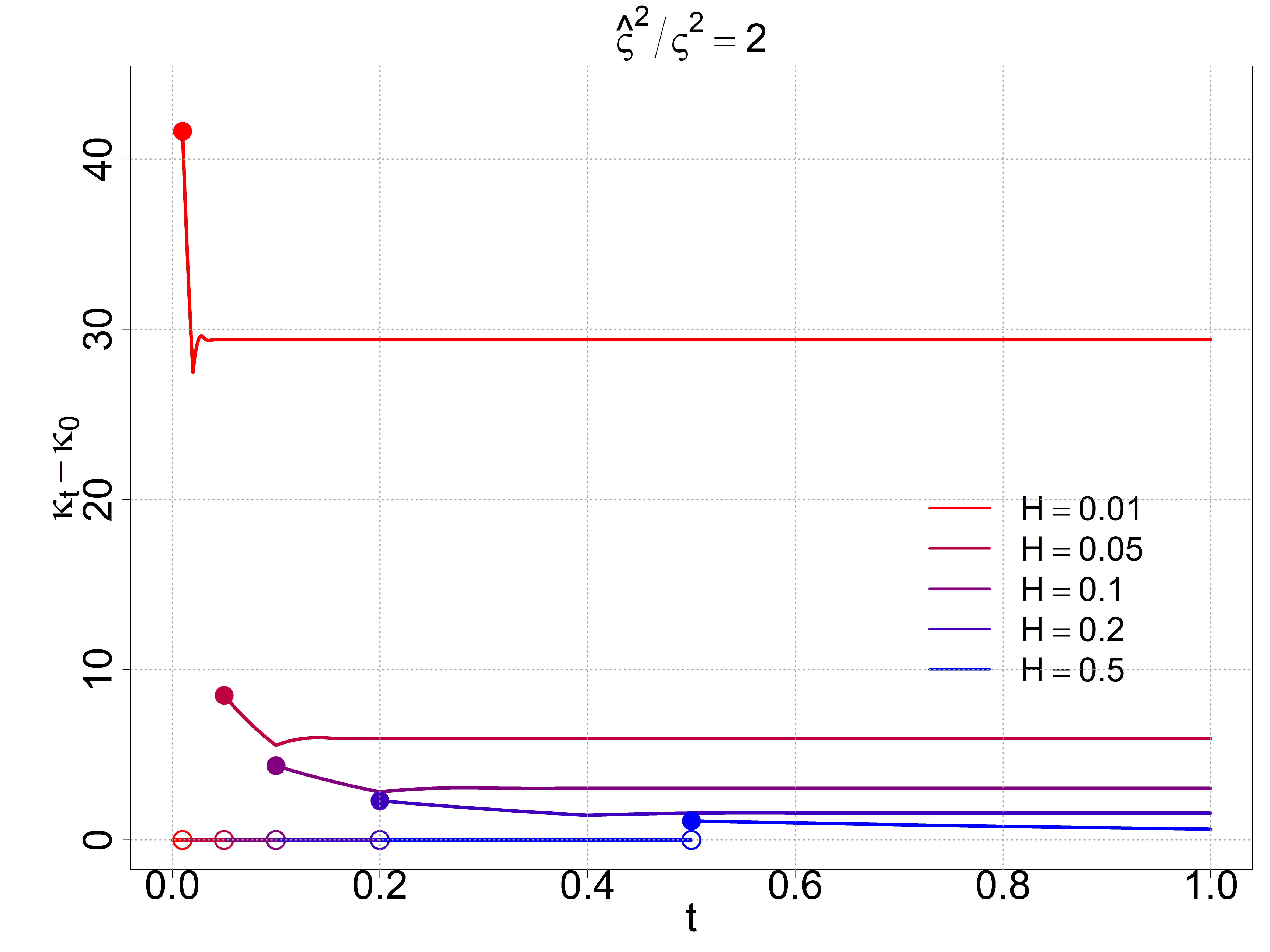}
\includegraphics[width=0.495\columnwidth]{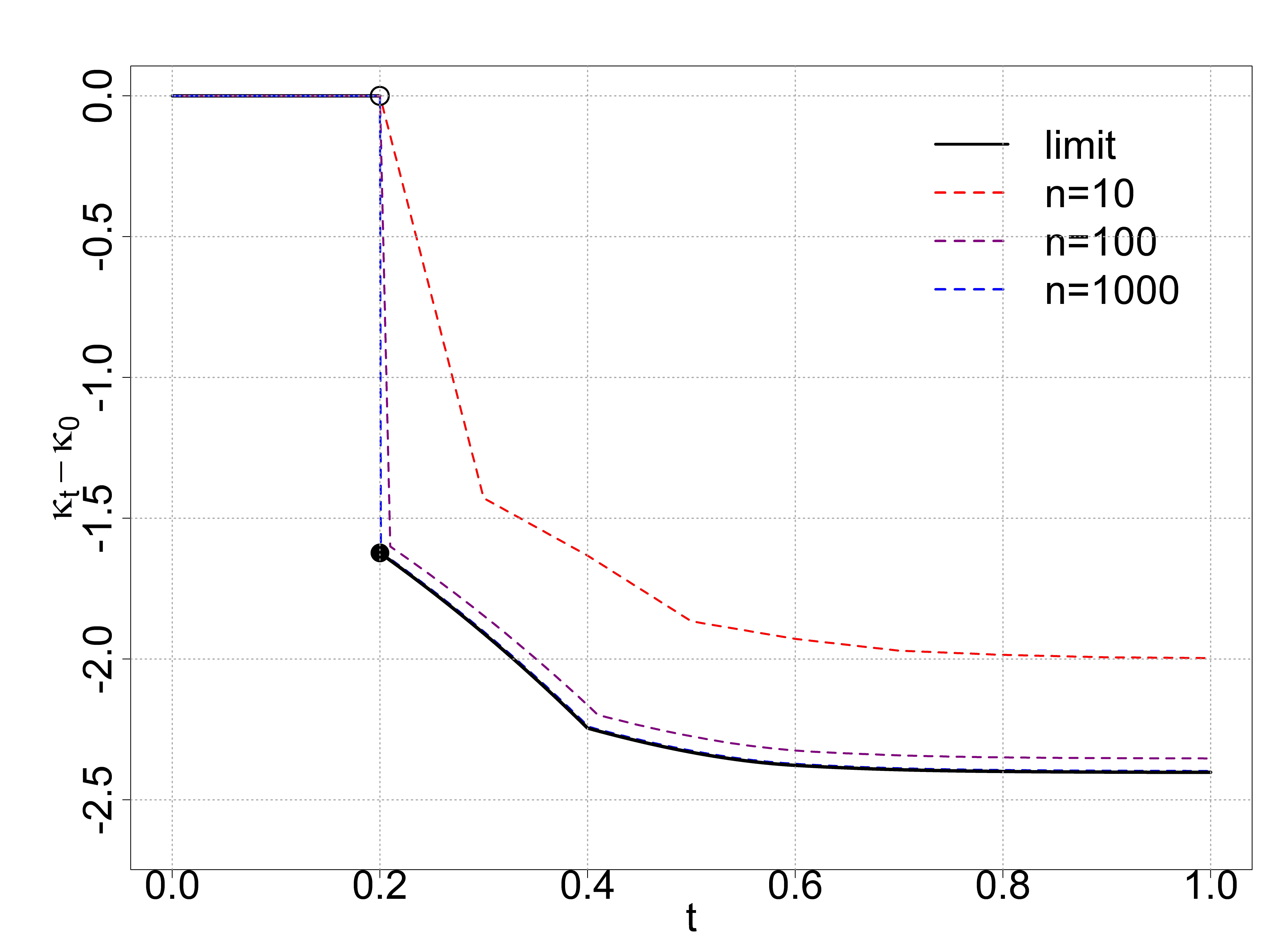}
\includegraphics[width=0.495\columnwidth]{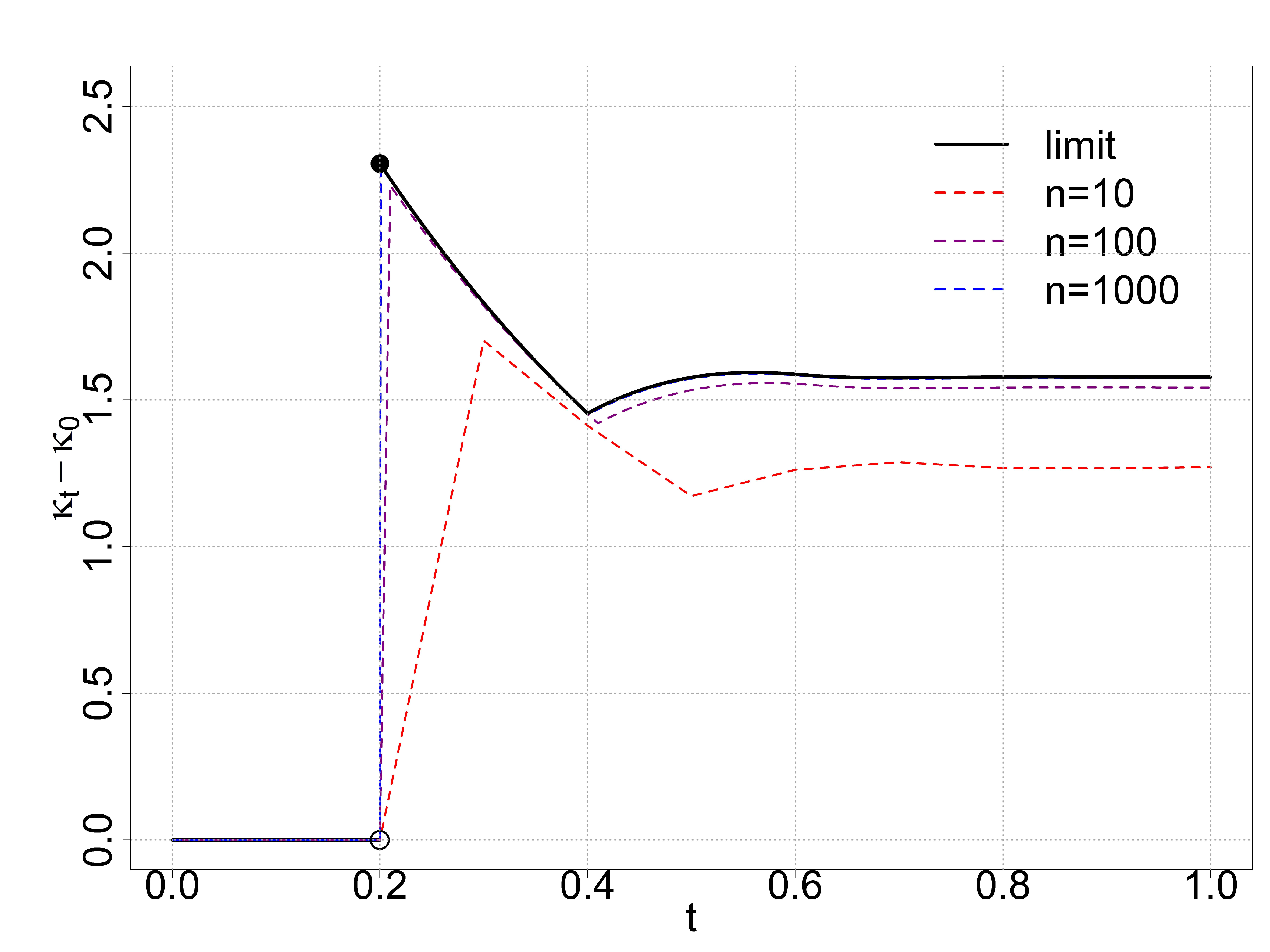}	\caption{\footnotesize 
Top: The scaled values of $\kappa_t-\kappa_0$ shown as a function of $t$ for $\frac{\hat \varsigma^2}{\varsigma^2} = 0.5$ (left) and $\frac{\hat \varsigma^2}{\varsigma^2}=2$ (right), computed for different values of $H$ (different colored curves). The full and empty circles represent the discontinuity point at $t=H$.
 For $\frac{\hat \varsigma^2}{\varsigma^2} < 1$, it holds that $0 < \alpha H < 1$ hence the investment weights $\kappa_t - \kappa_0$ are negative and monotonically decrease with $t$. 
 For $\frac{\hat \varsigma^2}{\varsigma^2} > 1$, we have $\alpha H < 0$ and the investment weights $\kappa_t - \kappa_0$ are positive.\\
Bottom: The scaled values $n b_i$ as a function of $\frac{i}{n}$ for $H=0.2$ and for $\frac{\hat \varsigma^2}{\varsigma^2} = 0.5$ (left) and $\frac{\hat \varsigma^2}{\varsigma^2}=2$ (right), computed for different values of $n$ (different dashed curves), together with the continuous curve $\kappa_t - \kappa_0$ shown as a function of $t$. As $n$ increases, the discrete values approach the continuous limit, and the curves are barely distinguishable for $n=1000$. For $\frac{\hat \varsigma^2}{\varsigma^2} = 0.5 < 1$ ($\frac{\hat \varsigma^2}{\varsigma^2} = 2 > 1$) the investment weights are negative (positive). The investment weights curves are zero for $t < H$ and show a discontinuity at $t = H$, with non-zero weights for $t \geq H$.
 \label{fig:strategy_limit_prop}}	
\end{figure}

 \begin{figure}[!ht]
 \centering
\includegraphics[width=0.95\columnwidth]{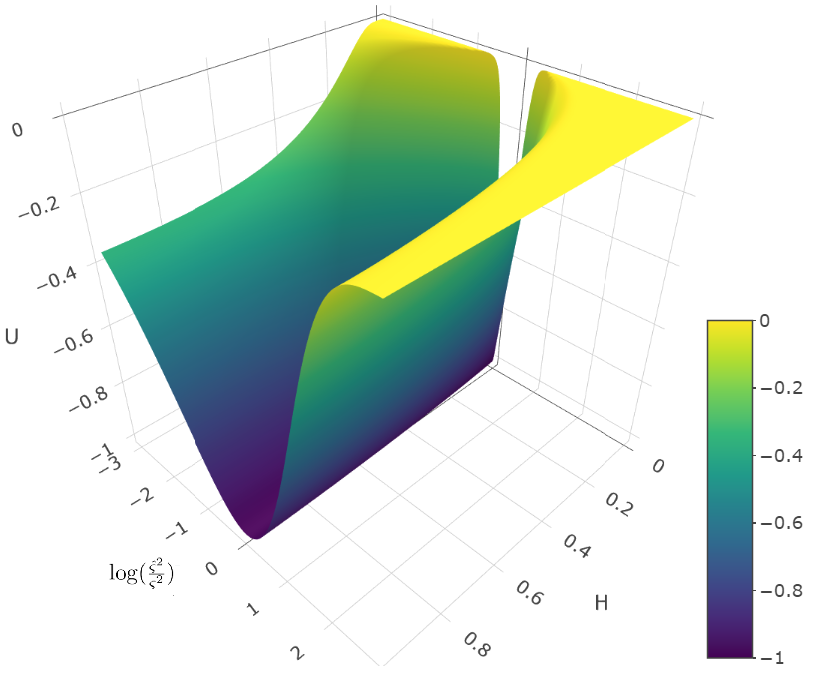}	\caption{\footnotesize 
The limit value $U(H,\varsigma,\hat\varsigma):=\lim_{n\rightarrow\infty} U_n(H,0,\varsigma,\hat\varsigma)$
(shown by color and the z-axis) for the case of zero drift ($\theta=0$)
depends only on the delay $H$ (x-axis) and 
$\log\left(\frac{\hat\varsigma}{\varsigma}\right)$ (y-axis).
We observe the following properties:\\
I. For the case where the market pricing rule 
is consistent with the physical measure (i.e. $\varsigma=\hat\varsigma$)
we have $\alpha=0$ and so $U=-1$.\\
II. If $H\downarrow 0$ then 
$\alpha H\rightarrow 1-\frac{\hat\varsigma}{\varsigma}$. Hence, for the case $\varsigma\neq \hat\varsigma$ it holds that
 $\alpha\left(\frac{\hat\varsigma^2}
{\varsigma^2(1-\alpha H)}+H-1\right)\rightarrow -\infty$ and so,
the right hand side of (\ref{4.0}) goes to $0$.
Namely, for vanishing delay $H\downarrow 0$ we have 
an arbitrage opportunity i.e. $U\rightarrow 0$ provided that there is 
inconsistency between the physical measure and the market pricing rule.
Financially speaking,
since the non delayed Bachelier model is complete, then the above inconsistency 
leads to an arbitrage opportunity.\\
III. Since $U(H,\varsigma,\hat\varsigma)$ is a limit value of discrete-time models, then from the discussion
in Section \ref{sec:2} it follows that for a given $H,\varsigma$
the function $U$ is increasing in $\hat\varsigma$ on the interval $[\varsigma,\infty)$ 
and decreasing in $\hat\varsigma$ on the interval $(0,\varsigma)$.
 Moreover, by using the 
the same arguments as in (\ref{static1})--(\ref{static2}) we obtain that 
if $\left |\log\left(\frac{\hat\varsigma}{\varsigma}\right)\right|\rightarrow\infty$ then $U\rightarrow 0$. 
}
 \label{fig:limit_value}	
\end{figure}

We end this section with the following remark.
\begin{rem}
\label{rem.new}
If $\hat{\varsigma}=\varsigma$ then $\alpha=0$ and so the optimal strategy is
$\Gamma\equiv \frac{\theta}{2\varsigma^2}$, $F\equiv 0$ (i.e. the Merton line)
and does not depend on the delay $H$. 
 For $\hat{\varsigma}\neq \varsigma$ we observe that the right-hand side of (\ref{4.0}) converges to zero as $H\downarrow 0$. In words, if $\hat{\varsigma}\neq \varsigma$ we obtain an asymptotic arbitrage for vanishing delay. Formally, in this case, $\alpha H \rightarrow 1$ (as $H\downarrow 0$)
 and so there is a singularity (blow-up) in the
optimal trading strategy $\pi_{*}=(\Gamma,F)$. 
In Figure \ref{fig:kappa_0_vs_H} and Figure \ref{fig:strategy_limit_prop} (Top) we provide additional (numerical) description of the optimal trading strategy $\pi_{*}=(\Gamma,F)$ for a few cases. 
 In addition, Figure \ref{fig:limit_value} displays the limiting value function from \eqref{4.0} as a function of the delay $H$ and the ratio $\frac{\hat\varsigma}{\varsigma}$ for the case of zero drift $(\theta=0)$. 
\end{rem}

\newpage
\clearpage

\section{Proofs for the Bachelier Model}
\subsection{Proof of Theorem \ref{thm4.1}}\label{sec:5}
 First, we state a simple discrete version of the Gr\"{o}nwall Lemma. Because we could not find a direct reference, we provide the proof for
 the reader's convenience. 
 \begin{lem}\label{lem5.1}
 Let a positive real sequence $\{z_j\}_{j=0}^n$ satisfy the recursive relation 
 \begin{equation}\label{5.rec}
 z_{j+1}\leq \phi z_j+\psi, \ \ j=0,1...,n-1
 \end{equation}
 for some constants $\phi>1$ and $\psi>0$. Then,
 $$z_j\leq \frac{\psi}{\phi-1} (\phi^j-1)+\phi^j z_0, \ \ \forall j. $$
 \end{lem}
 \begin{proof}
 The proof will be done by induction. For $j=0$ the statement is trivial. Assume the statement holds for $j$, and let us prove it for $j+1$. 
 Indeed, 
 \begin{align*}
 z_{j+1}&\leq \phi z_j+\psi\\
 & \leq \phi \left(\frac{\psi}{\phi-1} (\phi^j-1)+\phi^j z_0\right)+\psi\\
 &=\frac{\psi}{\phi-1} (\phi^{j+1}-1)+\phi^{j+1}z_0
 \end{align*}
 and the proof is completed. 
 \end{proof}

For a constant $C>0$ we use the notation $O\left(1/n^C\right)$ to denote a sequence $\{\alpha_n\}_{n=1}^{\infty}$ which satisfies
$\sup_{n\in\mathbb N} \alpha_n n^C<\infty$.
 Next, we prove the following auxiliary result.
 \begin{lem}\label{lem5.2}
For any $n\in\mathbb N$ introduce the function 
$b^n:[0,1]\rightarrow\mathbb R$ 
\begin{equation*}
b^n_t:=\begin{cases}
			n b_{k+1}, & \text{if \ $\frac{k}{n} \leq t <\frac{k+1}{n}$\ \ $k=0,1,...,n-2$}\\
 n b_n, & \text{if \ $\frac{n-1}{n}\leq t\leq 1$}
		 \end{cases}
\end{equation*}
where $\{b_i\}_{i=1}^{\infty}$ is given by (\ref{eq:a_sol})--(\ref{1.4})
for $D=D_n$, $\sigma=\frac{\varsigma}{\sqrt n}$ and $\hat\sigma=\frac{\hat\varsigma}{\sqrt n}$.

Then, 
$$\int_{0}^1 |b^n_t-\kappa_t|^2 dt = O(1/n).$$
\end{lem}
\begin{proof}
The proof will be done in three steps. \\
${}$\\
 \textbf{Step I:}
In this step we reduce the statement of the Lemma to a claim of discrete type.
Recall that the function $\kappa$ is constant on the interval $[0,H)$ and 
continuous on the interval $[H,1]$, hence from (\ref{eq:b_sol_recursive}) 
we conclude that $\kappa$ is Lipschitz continuous on the interval $[H,1]$.
Fix $n$. From the above, by applying the triangle inequality for the $L^2[0,1]$ norm we obtain (recall that $b^n$ is piecewise constant)
\begin{align*}
&\left(\int_{0}^1 |b^n_t-\kappa_t|^2 dt\right)^{\frac{1}{2}}\\
&\leq \left(\frac{1}{n}\sum_{k=0}^{n-1} \left|b^n_{\frac{k}{n}}-\kappa_{\frac{k+1}{n}}\right|^2\right)^{\frac{1}{2}}+
 \left(\sum_{k=0}^{n-1} \int_{\frac{k}{n}}^{\frac{k+1}{n}}\left|\kappa_{\frac{k+1}{n}}-\kappa_t\right|^2 dt\right)^{\frac{1}{2}}\\
 &=\left(\frac{1}{n}\sum_{k=0}^{n-1} \left|b^n_{\frac{k}{n}}-\kappa_{\frac{k+1}{n}}\right|^2\right)^{\frac{1}{2}}+O(1/\sqrt n).
\end{align*}
Since $\kappa$ is not necessarily continuous in $H$, then (in general)
$\int_{\frac{D_n-1}{n}}^{\frac{D_n}{n}}\left|\kappa_{\frac{D_n}{n}}-\kappa_t\right|^2 dt=O(1/n)$. This is the reason why 
 $$\left(\sum_{k=0}^{n-1} \int_{\frac{k}{n}}^{\frac{k+1}{n}}\left|\kappa_{\frac{k+1}{n}}-\kappa_t\right|^2 dt\right)^{\frac{1}{2}}=O(1/\sqrt n)$$
 and not $O(1/n)$.

We conclude that in order to prove the Lemma it is sufficient to establish that
$\max_{0\leq k< n}\left|b^n_{\frac{k}{n}}-\kappa_{\frac{k+1}{n}}\right|=O(1/n)$. This is equivalent to showing that
\begin{equation}\label{5.10} 
\Delta^n_i=O(1/n), \ \ \forall i=1,...,\lceil 1/H\rceil
\end{equation}
where 
$$\Delta^n_i :=\max_{(i-1)D_n\leq k<(i D_n)\wedge n}|b^n_{\frac{k}{n}}-\kappa_{\frac{k+1}{n}}|.$$
\\
${}$\\
 \textbf{Step II:}
In this step we derive some preliminary estimates. 
Let $a_n$ be the $a$ given by (\ref{eq:a_sol}) 
for $D=D_n$, $\sigma=\frac{\varsigma}{\sqrt n}$ and $\hat\sigma=\frac{\hat\varsigma}{\sqrt n}$.
Observe that, 
\begin{align}\label{4.4}
 a_n &:= \frac{\varsigma^2}{2n\hat\varsigma^2}+\frac{\sqrt{\left(2D_n-\frac{D^2_n \varsigma^2}{n\hat\varsigma^2}\right)^2- 4 D^2_n \left(1-\frac{\varsigma^2}{\hat\varsigma^2}\right)}-2D_n}{2 D^2_n}+O(1/n^2)\nonumber\\
 &= \left(\frac{\varsigma^2}{2\hat\varsigma^2}-\frac{1}{H}+
 \frac{\sqrt{4\frac{\varsigma^2}{\hat\varsigma^2}(1-H)+H^2\frac{\varsigma^4}{\hat\varsigma^4}}}{2 H}\right)\frac{1}{n}+O(1/n^2)\nonumber\\
 &= \frac{\alpha}{ (1-\alpha H) }\frac{1}{n}+O(1/n^2).
 \end{align}

Next, (\ref{1.4}) yields that for any $k\geq D_n-1$ 
$$b^n_{\frac{k+1}{n}}-b^n_{\frac{k}{n}}=\frac{a_n}{a_nD_n+1}\left(b^n_{\frac{k}{n}}-b^n_{\frac{k-D_n}{n}}\right).$$
Thus,
\begin{equation}\label{5.11}
b^n_{\frac{k+1}{n}}=\left(1+\frac{a_n}{a_nD_n+1}\right)b^n_{\frac{k}{n}}-
\frac{a_n}{a_nD_n+1}b^n_{\frac{k-D_n}{n}}, \ \ \forall k\geq D_n-1.
\end{equation}

From (\ref{prop}) we have that for any $k\geq D_n$
$$\kappa_{\frac{k+1}{n}}-\kappa_{\frac{k}{n}}=
\alpha\left(\int_{\frac{k}{n}}^{\frac{k+1}{n}}\kappa_s ds-
\int_{\frac{k}{n}-H}^{\frac{k+1}{n}-H}\kappa_s ds
\right).$$
This, together with (\ref{4.4})
and the fact that 
the function $\kappa$ is constant on the interval $[0,H)$ and Lipschitz continuous 
on the interval $[H,1]$
yield that for any $k\geq D_n$ such that $k\neq 2D_n-1,2 D_n-2$ 
\begin{align*}
\kappa_{\frac{k+1}{n}}-\kappa_{\frac{k}{n}}&=\frac{\alpha}{n}\left(\kappa_{\frac{k}{n}}-\kappa_{\frac{k-D_n}{n}}\right)+O(1/n^2)\\
&=\frac{a_n}{a_nD_n+1}\left(\kappa_{\frac{k}{n}}-\kappa_{\frac{k-D_n}{n}}\right)+O(1/n^2).
\end{align*}
We exclude $k=2D_n-1,2 D_n-2$ because for these $k$ the discontinuity point of $\kappa$ might lie between 
$\frac{k-D_n}{n}$ and $\frac{k+1}{n}-H$. Namely, 
$\frac{k-D_n}{n}\leq H\leq \frac{k+1}{n}-H$.

We conclude that for any $k\geq D_n$ such that $k\neq 2D_n-1,2 D_n-2$
\begin{equation}\label{5.12}
\kappa_{\frac{k+1}{n}}=\left(1+\frac{a_n}{a_nD_n+1}\right)\kappa_{\frac{k}{n}}-
\frac{a_n}{a_nD_n+1}\kappa_{\frac{k-D_n}{n}}+O(1/n^2), \ \ \forall k\geq D_n-1.
\end{equation}
\\
${}$\\
 \textbf{Step III:}
In this step we complete the proof. We establish (\ref{5.10}) by induction on $i$.
Observe that $\kappa_H=\frac{\alpha^2 H}{1-\alpha H}$ and 
$b^n_{\frac{D_n-1}{n}}=\frac{a^2_n D_n}{a_n D_n+1}$. Thus, 
from (\ref{4.4}) 
it follows that $\Delta^n_1=O(1/n)$. 
Assume that the statement holds true for $i$ and let us prove it for $i+1$. 
Define the sequence (which depends on $n$ and $i$)
$$z^{n,i}_j=\left|b^n_{\frac{i D_n+j-1}{n}}-\kappa_{\frac{ i D_n+j}{n}}\right|, \ \ j=0,1,...,D_n\wedge (n-i D_n).$$
Clearly, $z^{n,i}_0\leq \Delta_i$ and so from the induction assumption
$z^{n,i}_0= O(1/n)$.

From \eqref{5.11},\eqref{5.12} we obtain that for $k<D_n-2$.
\begin{align*}
z^{n,i}_{k+1} &= \left(1+\frac{a_n}{a_nD_n+1}\right)z^{n,i}_k+\frac{a_n}{a_nD_n+1}\left|b^n_{\frac{(i-1) D_n+k-1}{n}}-\kappa_{\frac{(i-1) D_n+k}{n}}\right|+O(1/n^2)\\
&=\left(1+\frac{a_n}{a_nD_n+1}\right)z^{n,i}_k+O(1/n^2)
\end{align*}
where the last equality follows from the fact that $\Delta^n_i=O(1/n)$ (induction assumption) and
$\frac{a_n}{a_nD_n+1}=O(1/n)$. 
We take $k<D_n-2$
only because of the case $i=1$ (recall 
that in (\ref{5.12}) we exclude $k=2D_n-1,2 D_n-2$).

Thus, the sequence $\{z^{n,i}_j\}_{j=0}^{D_n-2}$ satisfies the recursive relation 
 (\ref{5.rec}) for 
$\phi=1+\frac{a_n}{a_nD_n+1}=1+\frac{\alpha}{n}+O(1/n^2)$ and 
$\psi=O(1/n^2)$. Hence, from Lemma \ref{lem5.1} we obtain 
that 
\begin{equation}\label{5.14}
\max_{0\leq j\leq D_n-2} z^{n,i}_j=O(1/n).
\end{equation} 

It remains to treat the term 
$z^{n,i}_j$ for $j= D_n-2, D_n-1,D_n$.
Once again, since $\kappa$ is constant on the interval $[0,H)$ and Lipschitz continuous 
on the interval $[H,1]$, then from (\ref{4.4}) we have 
$$\kappa_{\frac{i D_n+j}{n}}=\frac{a_n}{a_n D_n+1}
\sum_{r=1}^{D_n} \kappa_{\frac{i D_n+j-r}{n}}+O(1/n), \ \ j= D_n-2, D_n-1,D_n.$$
This together with (\ref{1.4}), (\ref{5.14}) and the induction assumption $\Delta^n_i=O(1/n)$ gives 
$z^{n,i}_j=O(1/n)$ for $j= D_n-2, D_n-1,D_n$
and the proof is completed.
\end{proof}

Now, we are ready to prove Theorem \ref{thm4.1}.

\begin{proof}
 Choose $n\in\mathbb N$. We apply Theorem \ref{thm1.1} to the sequence 
 $S_k:=P_{k/n}$, $k=0,1...,n$. 
Thus, let $D=D_n$, $\mu=\frac{\theta}{n}$, $\sigma=\frac{\varsigma}{\sqrt n}$, $\hat\sigma=\frac{\hat\varsigma}{\sqrt n}$.
From (\ref{1.6}) and (\ref{4.4}) we obtain
\begin{align*}
&\lim_{n\rightarrow\infty} U_n(H,\mu,\varsigma,\hat\varsigma)\\
&=-\exp\left(\frac{1}{2}\left(-\frac{\theta^2}{\varsigma^2}+\lim_{n\rightarrow\infty}\left(na_n\frac{\hat\varsigma^2}{\varsigma^2}-\frac{a_n(n-D_n)}{1+a_n D_n}\right)\right)\right)\frac{1}{\sqrt{1+\lim_{n\rightarrow\infty} a_n D_n}}\\
&=-{\exp\left(\frac{1}{2}\left(-\frac{\theta^2}{\varsigma^2}+
\alpha\left(\frac{\hat\varsigma^2}
{\varsigma^2(1-\alpha H)}+H-1\right)\right)\right)}\sqrt{1-\alpha H}
\end{align*}
and (\ref{4.0}) follows. 

Next, by combining (\ref{1.5}) (for the function $f^{*}$) and (\ref{4.4}) we get (\ref{4.2}). 

Finally, we establish (\ref{4.3}).
Set 
$$\Phi_n=\max_{0\leq k\leq n-1}\left|\gamma^{*,n}_{\frac{k}{n}}-\Gamma_{\frac{k}{n}}\right|, \ \ n\in\mathbb N.$$
The function $\kappa$ is constant on the interval $[0,H)$ and Lipschitz continuous 
on the interval $[H,1]$ and so, from \cite{Fernique}
we obtain that the Volterra Gaussian process $\{\Gamma_t\}_{t=0}^1$ is continuous. 
Hence (recall that $\gamma^{*,n}$ is piecewise constant), in order to establish (\ref{4.3}) it is sufficient to show that 
$\lim_{n\rightarrow\infty}\Phi_n=0$ a.s. 

In view of the Borel–Cantelli lemma
 this will follow from the inequality
$$\sum_{n=1}^{\infty}\mathbb P\left(\Phi_n>\epsilon\right)<\infty, \ \ \forall \epsilon>0.$$

Clearly, the Markov inequality (for $\Phi^6_n$) implies that 
$$\sum_{n=1}^{\infty}\mathbb P\left(\Phi_n>\epsilon\right)\leq \frac{1}{\epsilon^6}
\sum_{n=1}^{\infty} \mathbb E_{\mathbb P}\left[\Phi^6_n\right].$$
Thus, it remains to establish the inequality
$\sum_{n=1}^{\infty} \mathbb E_{\mathbb P}\left[\Phi^6_n\right]<\infty.$

To this end we apply Lemma \ref{lem5.2}. Fix $n$. 
By applying (\ref{1.5}) for 
$S_k:=P_{k/n}$, $k=0,1...,n$ it follows that 
$$\gamma^{*,n}_{\frac{k}{n}}=\frac{\theta}{\varsigma^2}+\frac{1}{\varsigma^2}\int_{0}^{k/n}\left(b^n_{\frac{k}{n}-s}-na_n\right)dP_s$$
where $b^n$ and $a_n$ are as before. 
Thus, from (\ref{vol}), Lemma \ref{lem5.2} and the fact that $\int_{t_1}^{t_2}\zeta_s dw_s\sim \mathcal N\left(0,\int_{t_1}^{t_2}\zeta^2_s ds\right)$ for a deterministic $\zeta$ we obtain
\begin{align*}
\mathbb E_{\mathbb P}\left[\Phi^6_n\right] &\leq \sum_{k=0}^{n-1} \mathbb E_{\mathbb P}\left[\left|\gamma^{*,n}_{\frac{k}{n}}-\Gamma_{\frac{k}{n}}\right|^6\right]\\
&=\frac{1}{\varsigma^2}\sum_{k=0}^{n-1} \mathbb E_{\mathbb P}\left[\left|\int_{0}^{k/n} \left(\left(b^n_{\frac{k}{n}-s}-na_n\right)-\left(\kappa_{\frac{k}{n}-s}-\frac{\alpha}{1-\alpha H}\right)\right)dP_s\right|^6\right]\\
&= n O(1/n^3)=O(1/n^2)
\end{align*}
and the proof is completed. 
\end{proof} 

\subsection{Proof of the Upper Bound (\ref{9.2})}\label{sec:6}
Let $\{\mathcal F_t\}_{t=0}^1$, $t\in [0,1]$ be the augmented filtration which is generated by $\{W_t\}_{t=0}^1$.
We start with 
the following key lemma. 
\begin{lem}\label{lem6.1}
Let $\epsilon>0$ and $\gamma=(\gamma_t)_{t\in [0,1]}$ be a predictable process with respect to the filtration 
$\left(\mathcal F_{(t-\epsilon)^{+}}\right)_{t\in [0,1]}$ which satisfies $\int_{0}^1\gamma^2_t dt<\infty$ a.s.
If the expectation of the negative part satisfies $\mathbb E_{\mathbb P}\left[\left(\int_{0}^1 \gamma_t dW_t\right)^{-}\right]<\infty$ then 
$\mathbb E_{\mathbb P}\left[\int_{0}^1 \gamma_t dW_t\right]=0.$
 \end{lem}
\begin{proof}
Choose $n\in\mathbb N$ such that $\frac{1}{n}<\epsilon$. 
Let us prove by backward induction that for any $k=0,1,...,n$ 
\begin{equation}\label{6.2}
\mathbb E_{\mathbb P}\left[\int_{0}^{1} \gamma_t dW_t\left.|\right.\mathcal F_{\frac{k}{n}}\right]=
\int_{0}^{\frac{k}{n}}\gamma_t dW_t.
\end{equation}
For $k=n$ (\ref{6.2}) is obvious. Assume that (\ref{6.2}) is correct for $k$, let us prove it for $k-1$. 
From the Jensen inequality, 
the fact that $\mathbb E_{\mathbb P}\left[\left(\int_{0}^1 \gamma_t dW_t\right)^{-}\right]<\infty$ and the induction assumption 
it follows that $\mathbb E_{\mathbb P}\left[\left(\int_{0}^{\frac{kT}{n}} \gamma_t dW_t\right)^{-}\right]<\infty$ hence the term
$\mathbb E_{\mathbb P}\left[\int_{0}^{\frac{k}{n}} \gamma_t dW_t\left|\right. \mathcal F_{\frac{k-1}{n}}\right]$ is well defined. 
Moreover, the tower property and the induction assumption imply
$$\mathbb E_{\mathbb P}\left[\int_{0}^{1} \gamma_t dW_t\left|\right.\mathcal F_{\frac{k-1}{n}}\right]
=\mathbb E_{\mathbb P}\left[\int_{0}^{\frac{k}{n}} \gamma_t dW_t\left.|\right.\mathcal F_{\frac{k-1}{n}}\right].$$
Thus, in order to prove (\ref{6.2}) for $k-1$ it remains to establish that 
\begin{equation}\label{past}
\mathbb E_{\mathbb P}\left[\int_{\frac{k-1}{n}}^{\frac{k}{n}} \gamma_t dW_t\left|\right. \mathcal F_{\frac{k-1}{n}}\right]=0.
\end{equation}
Indeed, since $\frac{1}{n}<\epsilon$ then $(\gamma_t)_{t\in \left[\frac{k-1}{n},\frac{k}{n}\right]}$ is $\mathcal F_{\frac{k-1}{n}}$ measurable, 
and so, the regular conditions distribution (for details see \cite{shiryaev2016probability}) of $\int_{\frac{k-1}{n}}^{\frac{k}{n}} \gamma_t dW_t$ with respect to 
$\mathcal F_{\frac{k-1}{n}}$ is the normal distribution $\mathcal N\left(0,\int_{\frac{k-1}{n}}^{\frac{k}{n}} \gamma^2_t dt\right)$.
Thus, by using (a priori) generalized conditional expectations (as defined in \cite{shiryaev2016probability}) we obtain 
(\ref{past}).
This completes the proof of (\ref{6.2}). Taking $k=0$ in (\ref{6.2}) we obtain 
$\mathbb E_{\mathbb P}\left[\int_{0}^1 \gamma_t dW_t\right]=0.$
\end{proof}
\begin{rem}
Let us notice that for $\epsilon=0$ Lemma \ref{lem6.1} does not hold true. Indeed, 
the doubling strategies from \cite{harrison1979martingales} are piecewise constant with a (random) finite number of trading times
such that the corresponding portfolio value at the maturity date is equal to $1$ almost surely. 
\end{rem}

Next, we will need some preparation. 
It is well known (see \cite{H:68}) that for any 
deterministic and measurable function $l_{t,s}$, $0\leq s\leq t\leq 1$ (Volterra kernel) 
belonging to $L^2([0,1]^2)$
there exists a probability measure 
$\mathbb Q(l)\sim\mathbb P$ 
and a process $W^{\mathbb Q(l)}=(W^{\mathbb Q(l)}_t)_{t\in [0,1]}$ such that 
$W^{\mathbb Q(l)}$ is a $\mathbb Q(l)$-Brownian motion and we have the equality
\begin{equation}\label{gir}
W^{\mathbb Q(l)}_t=W_t+\int_{0}^t\int_{0}^s l_{s,u} dW^{\mathbb Q(l)}_u ds, \ \ t\in [0,1].
\end{equation}
From the Girsanov theorem 
$$\frac{d\mathbb Q(l)}{d\mathbb P}=\exp\left(-\int_{0}^1 \int_{0}^t l_{t,s} dW^{\mathbb Q(l)}_sdW^{\mathbb Q(l)}_t+
\frac{1}{2}\int_{0}^1\left(\int_{0}^t l_{t,s} dW^{\mathbb Q(l)}_s\right)^2 dt\right). $$
Hence, 
from the It\^{o} Isometry and the Fubini theorem
we obtain that the relative entropy is given by 
\begin{equation}\label{2+.entropy}
\mathbb E_{\mathbb Q(l)}\left[\log\left(\frac{d\mathbb Q(l)}{d\mathbb P}\right)\right]
=\frac{1}{2}\int_{0}^1\int_{0}^t l^2_{t,s}ds dt=\frac{1}{2}\int_{0}^1\int_{s}^1 l^2_{t,s}dt ds.
\end{equation}
For any $\epsilon>0$ denote by $\Gamma_{\epsilon}$ the set of all deterministic, measurable and bounded functions $l_{t,s}$, $0\leq s\leq t\leq 1$ 
such that $l_{t,s}=0$ if $t>s+H-\epsilon $ and
\begin{equation}\label{variance}
\int_{0}^1\left(1-\int_{s}^1 l_{t,s} dt\right)^2 ds=\hat{\varsigma}^2.
\end{equation}

Now, we are ready to prove (\ref{9.2}).

\begin{proof}
The proof will be done in three steps. 
\\
${}$\\
 \textbf{Step I:}
 We start with reducing the proof to the special case where $P_0=\theta=0$ and $\varsigma=1$,
i.e. $P=W$. By considering the map $P\rightarrow \frac{P-P_0}{\varsigma}$ and 
$\mathcal N(P_0, \hat{\varsigma}^2)\rightarrow \mathcal N\left(0, \frac{\hat{\varsigma}^2}{\varsigma^2}\right)$
it follows that if (\ref{9.2}) holds true for the case
$P_0=0$ and $\varsigma=1$ then it holds for the general case. Thus, it remains to treat the drift $\theta$.
Assume that $P_0=0$, $\varsigma=1$ and 
introduce the probability measure 
$\tilde{\mathbb P}$ by 
$\frac{d\tilde{\mathbb P}}{d\mathbb P}:=\exp\left(-\theta P_T+\frac{\theta^2}{2}\right).$
From the Girsanov theorem it follows that $P$ is a standard Brownian motion under 
$\tilde{\mathbb P}$.
This property together with the simple equality 
$$\mathbb E_{\mathbb P}\left[-\exp\left(-V^{{\gamma,f}}_1\right)\right]=
 \exp\left(-\frac{\theta^2}{2} \right)\mathbb E_{\tilde{\mathbb P}}\left[-\exp\left(-V^{{\gamma-\theta,f}}_T\right)\right] 
 \ \ \forall (\gamma,f)\in\mathfrak A
$$
yield that if Theorem (\ref{9.2}) holds true for $P=W$ then it holds true for the case of constant drift as well.
Hence from now on we assume that $P=W$. 
\\
${}$\\
\textbf{Step II:}
Let $\epsilon>0$. In this step we prove that for any 
$l\in\Gamma_{\epsilon}$ and 
$\pi\in\mathfrak A$ with $\mathbb E_{\mathbb P}\left[\exp\left(-V^{ \pi}_1\right)\right]<\infty$ we have
$\mathbb E_{\mathbb Q(l)}[V^{\pi}_1]=0.$

Choose $l\in\Gamma_{\epsilon}$, 
$\pi\in\mathfrak A$ with $\mathbb E_{\mathbb P}\left[\exp\left(-V^{ \pi}_1\right)\right]<\infty$
 and denote $\mathbb Q:=\mathbb Q(l)$.
From the Fubini theorem and (\ref{gir}) it follows that 
$W_1=\int_{0}^1\left(1-\int_{s}^1 l_{t,s}dt\right)dW^{\mathbb Q}_s.$
This together with (\ref{variance}) imply 
that $(W_1;\mathbb Q) \sim \hat{\mathcal V}$.
Let us show 
that $\mathbb E_{\mathbb Q}\left[\left(\int_{0}^1 \gamma_t dW_t\right)^{-}\right]<\infty$.
Indeed, in view of the relations 
$\int f d \hat{\mathcal V}<\infty$ and $(W_1;\mathbb Q) \sim \hat{\mathcal V}$ we have 
$f(W_T)\in L^1(\mathbb Q)$ and 
\begin{equation}\label{6+.3}
\mathbb E_{\mathbb Q} [f(W_1)]=\int f d \hat{\mathcal V}. 
\end{equation}

From the classical Legendre-Fenchel duality inequality
$xy \leq e^x + y(\log y-1)$ for $x=\left(V^{\pi}_1\right)^{-}$ 
and $y=\frac{d\mathbb Q}{d\mathbb P}$ we obtain 
$
\mathbb E_{\mathbb Q}\left[\left(V^{\pi}_1\right)^{-}\right]\leq 
\mathbb E_{\mathbb P}\left[\exp\left(-V^{ \pi}_1\right)\right]+
\mathbb E_{\mathbb Q}\left[\log\left(\frac{d\mathbb Q}{d\mathbb P}\right)\right]
<\infty.
$
This together with 
(\ref{6+.3}) yield
$\mathbb E_{\mathbb Q}\left[\left(\int_{0}^1 \gamma_t dW_t\right)^{-}\right]<\infty$.
Finally, 
$$\mathbb E_{\mathbb Q}\left[\int_{0}^1\gamma_t dW_t\right]=
\mathbb E_{\mathbb Q}\left[\int_{0}^1\left(\gamma_s-\int_{s}^{s\vee\left(s+H-\epsilon\right)}\gamma_t l_{t,s}dt\right) dW^{\mathbb Q}_s\right]=0.$$
The first equality follows from the Fubini theorem and the fact that $l\in\Gamma_{\epsilon}$.
 The second equality is argued by observing that the process
$\left(\gamma_s-\int_{s}^{s\vee\left(s+H-\epsilon\right)}\gamma_tl_{t,s}dt\right)$, $s\in [0,1]$
 is a predictable process with respect to the filtration 
 $\mathcal F_{(s-\epsilon)^+}$ and applying Lemma \ref{lem6.1}.
 This completes the second step. 
\\
${}$\\
\textbf{Step III:}
First, in view of \textbf{Step II}, by applying the same arguments as after 
(\ref{2.6}) we obtain that for any $\epsilon>0$
\begin{equation}\label{6.4}
\log\left(\mathbb E_{\mathbb P}\left[\exp\left(-V^{\pi}_1\right)\right] \right)\geq -\mathbb E_{\mathbb Q(l)}\left[\log\left(\frac{d\mathbb Q(l)}{d\mathbb P}\right)\right] \ \ \forall (\pi,l)\in\mathfrak A\times\Gamma_{\epsilon}.
\end{equation}
Set 
$\hat a:=\frac{\alpha}{1-\alpha H}$
where $\alpha$ is given by \eqref{eq:alpha_def}. Direct computations yield 
\begin{equation}\label{comp}
(1-\alpha H)(1+\hat a H)=1 \ \ \mbox{and} \ \ 
\frac{1+\hat aH^2}{(1+\hat a H)^2}=\frac{\hat{\varsigma}^2}{\varsigma^2}.
\end{equation}
Next, for any $\epsilon>0$ define $l^{\epsilon}\in\Gamma_{\epsilon}$ by
$l^{\epsilon}_{t,s}:=
\frac{\hat a\mathbb I_{ t<s+H-\epsilon}}{1+\hat a\left(H\wedge (1-s)\right)}$, $0\leq s\leq t\leq 1.$
By combining (\ref{2+.entropy}) and (\ref{6.4}) and taking $\epsilon\downarrow 0$ we conclude that 
\begin{equation}\label{6.5}
\sup_{\pi\in\mathfrak A}\mathbb E_{\mathbb P}\left[-\exp\left(-V^{\pi}_1\right)\right] \leq -\exp\left(
\frac{1}{2} \int_{0}^1 \frac{\hat a^2\left(H\wedge(1-s)\right)}{\left(1+\hat a\left(H\wedge(1-s)\right)\right)^2}ds\right).
\end{equation}
Observe that
\begin{align}\label{6.6}
&\int_{0}^1 \frac{\hat a^2\left(H\wedge(1-s)\right)}{\left(1+\hat a\left(H\wedge(1-s)\right)\right)^2}ds\nonumber\\
 &=\frac{\hat a^2 H (1-H)}{\left(1+\hat a H\right)^2}+\hat a^2\int_{1-H}^1 \frac{1-s}{\left(1+\hat a (1-s)\right)^2}ds\nonumber\\
 &= \frac{\hat a}{1+\hat a H}\frac{\hat aH(1-H)}{1+\hat a H}
 -1+\frac{1}{1+\hat aH}+ \log(1+\hat a H)\nonumber\\
 &=\alpha \left(1-\frac{\hat\varsigma^2}{\varsigma^2 (1-\alpha H)}\right)-\alpha H
 -\log(1-\alpha H)
 \end{align}
where the first two equalities are obtained by simple computations and the 
last equality follows from (\ref{comp}).
Finally, by combining (\ref{6.5})-(\ref{6.6}) we obtain (\ref{9.2}) and complete the proof. 
\end{proof}

\appendix
\section{Computational Details} \label{sec:appendix}

We provide a \href{https://github.com/orzuk/BandedDecomposition}{github} repository with freely available $R$ code that allows to numerically compute the value $U_n$ and its continuous limit, as well as the strategy weights $b_i$ and the continuous limit curve $\kappa$ as functions of the problems parameters, all implemented as $R$ functions in the file {\it `SemiStatic.R`}. The code for generating the figures shown in the paper is in the file {\it `RunSemiStatic.R`}.

The first $5$ coefficients $c_k$ defined in \eqref{eq:c_coefficients_recursive}, computed using the symbolic Maple program {\it `ContinuousLimitRecurrence.mw`} available in the repository, are shown below. 

{\footnotesize
\begin{align*}
c_1 &= -\alpha \\
c_2 &=-{\mathrm e}^{ \alpha H} \alpha \\
c_3 &= {\mathrm e}^{ \alpha H} \alpha ( \alpha H -{\mathrm e}^{ \alpha H}) \\
c_4 &= \frac{{\mathrm e}^{ \alpha H} \alpha (-H^{2} \alpha^{2}+4 \,{\mathrm e}^{ \alpha H} H \alpha -2 \,{\mathrm e}^{2 \alpha H})}{2} \\
c_5 &= \frac{{\mathrm e}^{ \alpha H} (-6 \,{\mathrm e}^{3 \alpha H}+(18 \,{\mathrm e}^{2 \alpha H}+ \alpha H ( \alpha H -12 \,{\mathrm e}^{ \alpha H})) \alpha H) \alpha}{6} 
\end{align*}
}

\newpage
\clearpage

\bibliographystyle{siam}
\bibliography{finance}

\begin{thebibliography}{10}

\bibitem{acciaio2016model}
{\sc B.~Acciaio, M.~Beiglb{\"o}ck, F.~Penkner, and W.~Schachermayer}, {\em A
  model-free version of the fundamental theorem of asset pricing and the
  super-replication theorem}, Mathematical Finance, 26 (2016), pp.~233--251.

\bibitem{backhoff2023most}
{\sc J.~Backhoff-Veraguas and M.~Beiglboeck}, {\em The most exciting game},
  Electronic Communications in Probability, 29 (2024), pp.~1--12.

\bibitem{BD:2021}
{\sc P.~Bank and Y.~Dolinsky}, {\em A note on utility indifference pricing with
  delayed information}, SIAM Journal on Financial Mathematics, 12 (2021),
  pp.~SC--31--SC--43.

\bibitem{BF:81}
{\sc W.~Barrett and P.~Feinsilver}, {\em Inverses of banded matrices}, Linear
  Algebra and its Applications, 41 (1981), pp.~111--130.

\bibitem{bartl2019exponential}
{\sc D.~Bartl}, {\em Exponential utility maximization under model uncertainty
  for unbounded endowments}, The Annals of Applied Probability, 29 (2019),
  pp.~577--612.

\bibitem{bayraktar2016arbitrage}
{\sc E.~Bayraktar and Z.~Zhou}, {\em Arbitrage, hedging and utility
  maximization using semi-static trading strategies with american options}, The
  Annals of Applied Probability, 26 (2016), pp.~3531--3558.

\bibitem{bouchard2015arbitrage}
{\sc B.~Bouchard and M.~Nutz}, {\em Arbitrage and duality in nondominated
  discrete-time models}, The Annals of Applied Probability, 25 (2015),
  pp.~823--859.

\bibitem{cartea2023optimal}
{\sc {\'A}.~Cartea and L.~S{\'a}nchez-Betancourt}, {\em Optimal execution with
  stochastic delay}, Finance and Stochastics, 27 (2023), pp.~1--47.

\bibitem{cox2011robust}
{\sc A.~M. Cox and J.~Ob{\l}{\'o}j}, {\em Robust pricing and hedging of double
  no-touch options}, Finance and Stochastics, 15 (2011), pp.~573--605.

\bibitem{dolinsky2014martingale}
{\sc Y.~Dolinsky and H.~M. Soner}, {\em Martingale optimal transport and robust
  hedging in continuous time}, Probability Theory and Related Fields, 160
  (2014), pp.~391--427.

\bibitem{dolinsky2023exponential}
{\sc Y.~Dolinsky and O.~Zuk}, {\em Exponential utility maximization in a
  discrete time gaussian framework}, SIAM Journal on Financial Mathematics., 14
  (2023), pp.~SC--41--SC--47.

\bibitem{fahim2016model}
{\sc A.~Fahim and Y.-J. Huang}, {\em Model-independent superhedging under
  portfolio constraints}, Finance and Stochastics, 20 (2016), pp.~51--81.

\bibitem{Fernique}
{\sc X.~Fernique}, {\em Continuité des processus gaussiens}, C. R. Acad. Sci.
  Paris, 258 (1964), pp.~6058--6060.

\bibitem{guo2017tightness}
{\sc G.~Guo, X.~Tan, and N.~Touzi}, {\em Tightness and duality of martingale
  transport on the skorokhod space}, Stochastic Processes and their
  Applications, 127 (2017), pp.~927--956.

\bibitem{harrison1979martingales}
{\sc J.~M. Harrison and D.~M. Kreps}, {\em Martingales and arbitrage in
  multiperiod securities markets}, Journal of Economic theory, 20 (1979),
  pp.~381--408.

\bibitem{H:68}
{\sc M.~Hitsuda}, {\em Representation of gaussian processes equivalent to
  wiener process}, Osaka Journal of Mathematics, 5 (1968), pp.~299--312.

\bibitem{hobson1998robust}
{\sc D.~G. Hobson}, {\em Robust hedging of the lookback option}, Finance and
  Stochastics, 2 (1998), pp.~329--347.

\bibitem{KX:07}
{\sc M.~Kohlmann and D.~Xiong}, {\em The mean-variance hedging of a defaultable
  option with partial information}, Stochastic Analysis and Applications, 25
  (2007), pp.~869--893.

\bibitem{mania2008mean}
{\sc M.~Mania, R.~Tevzadze, and T.~Toronjadze}, {\em Mean-variance hedging
  under partial information}, SIAM Journal on Control and Optimization, 47
  (2008), pp.~2381--2409.

\bibitem{Pham:01}
{\sc H.~Pham and M.~Quenez}, {\em Optimal portfolio in partially observed
  stochastic volatility models}, Annals of Applied Probability, 11 (2001),
  pp.~210--238.

\bibitem{rodman1992inversion}
{\sc L.~Rodman and T.~Shalom}, {\em On inversion of symmetric toeplitz
  matrices}, SIAM journal on matrix analysis and applications, 13 (1992),
  pp.~530--549.

\bibitem{saporito2019stochastic}
{\sc Y.~F. Saporito and J.~Zhang}, {\em Stochastic control with delayed
  information and related nonlinear master equation}, SIAM Journal on Control
  and Optimization, 57 (2019), pp.~693--717.

\bibitem{schweizer1994risk}
{\sc M.~Schweizer}, {\em Risk-minimizing hedging strategies under restricted
  information}, Mathematical Finance, 4 (1994), pp.~327--342.

\bibitem{schweizer2018dynamic}
{\sc M.~Schweizer, D.~Zivoi, and M.~{\v{S}}iki{\'c}}, {\em Dynamic
  mean--variance optimization problems with deterministic information},
  International Journal of Theoretical and Applied Finance, 21 (2018),
  p.~1850011.

\bibitem{shiryaev2016probability}
{\sc A.~N. Shiryaev}, {\em Probability-1}, vol.~95, Springer, 2016.

\end{thebibliography}

\end{document}